\documentclass[AMA,STIX1COL]{WileyNJD-v2}

\usepackage{graphicx,bm,inputenc,xcolor}
\usepackage{fancyhdr}
\usepackage{amsmath}
\usepackage{array}

\newcommand{\R}{\mathbb{R}}
\newcommand{\B}{\mathcal{B}}
\newcommand{\F}{\mathcal{F}}
\newcommand{\G}{\mathcal{G}}
\newcommand{\K}{\mathcal{K}}
\newcommand{\Q}{\mathcal{Q}}

\newcommand{\Kinf}{\mathcal{K}_{\infty}}
\newcommand{\Keinf}{\mathcal{K}^{\rm e}_{\infty}}

\articletype{Article Type}%

\received{26 April 2016}
\revised{6 June 2016}
\accepted{6 June 2016}

\begin{document}

\title{
% Safety-Critical Control Design for Time Delay Systems Using Control Barrier Functionals
Control Barrier Functionals: \\ Safety-critical Control for Time Delay Systems
% \protect\thanks{This is an example for title footnote.}
}

\author[1]{Adam K. Kiss*}

\author[2]{Tamas G. Molnar%, \IEEEmembership{Member, IEEE}
}

\author[2]{Aaron D. Ames%, \IEEEmembership{Fellow, IEEE}
}
\author[3]{G\'abor Orosz%, \IEEEmembership{Member, IEEE}
}

\authormark{AK KISS \textsc{et al}}

\address[1]{\orgdiv{MTA-BME Lend{\"{u}}let Machine Tool Vibration Research Group, Department of Applied Mechanics}, \orgname{Budapest University of Technology and Economics}, \orgaddress{\state{Budapest 1111}, \country{Hungary}}}

\address[2]{\orgdiv{Department of Mechanical and Civil Engineering}, \orgname{California Institute of Technology}, \orgaddress{\state{Pasadena, CA 91125}, \country{USA}}}

\address[3]{\orgdiv{Department of Mechanical Engineering and Department of Civil and Environmental Engineering}, \orgname{University of Michigan}, \orgaddress{\state{Ann Arbor, MI 48109}, \country{USA}}}

\corres{*Adam K. Kiss, MTA-BME Lend{\"{u}}let Machine Tool Vibration Research Group, Department of Applied Mechanics, Budapest University of Technology and Economics, Budapest 1111, Hungary. \email{kiss\_a@mm.bme.hu}}

% \presentaddress{This is sample for present address text this is sample for present address text}

\abstract[Summary]{

This work presents a theoretical framework for the safety-critical control of time delay systems.
The theory of control barrier functions, that provides formal safety guarantees for delay-free systems, is extended to systems with state delay.
The notion of control barrier functionals is introduced to attain formal safety guarantees, by enforcing the forward invariance of safe sets defined in the infinite dimensional state space.
The proposed framework is able to handle multiple delays and distributed delays both in the dynamics and in the safety condition, and provides an affine constraint on the control input that yields provable safety.
This constraint can be incorporated into optimization problems to synthesize pointwise optimal and provable safe controllers.
The applicability of the proposed method is demonstrated by numerical simulation examples.
}


\jnlcitation{\cname{%
\author{A. K. Kiss}, 
\author{T. G. Molnar}, 
\author{A. D. Ames} and 
\author{G. Orosz}
} (\cyear{2021}), 
\ctitle{
Control Barrier Functionals: Safety-critical Control for Time Delay Systems
}, \cjournal{Int J Robust Nonlinear Control}, \cvol{2021;XX:X-X}.}
\maketitle

%%%%%%%%%%%%%%%%%%%%%%%%%%%%%%%%%%%%%%%%%%%%%%%%%%%%%%%%%%%%%%%%%%%%%%%%%%%%%%%%%%%%%%%%%%%%
\section{Introduction}
%%%%%%%%%%%%%%%%%%%%%%%%%%%%%%%%%%%%%%%%%%%%%%%%%%%%%%%%%%%%%%%%%%%%%%%%%%%%%%%%%%%%%%%%%%%%

In modern control systems, safety is a crucial factor -- often a necessary precursor to other control objectives including: performance, efficiency and sustainability. This motivates the importance of developing safety-critical control methods. % Hence, safety considerations are extensively researched and increasingly important.
The application domains thereof are wide-spread, from self-driving autonomous vehicles \cite{Niletal2016}, through robotic systems \cite{tordesillas2019faster, Kousik2020, Nubert2020} to human-robot collaboration \cite{Zanchettin2016, Landi2019, Singletary2019} where safety plays a key role for reliable autonomy or sustainable operation. The importance of safety reaches even beyond engineering applications: for example,
the requirement of safety appears in epidemiological models describing pandemics \cite{ames2020safety,Molnar2021lcss} and in other biological applications. 

To formally address safety in dynamical and control systems, one can define a safe set over the state space, wherein safety can be framed  as the forward invariance of that set, i.e., one requires that the system evolves within the safe set for all time.
Rigorous guarantees of safety necessitate a theory for ensuring forward set invariance.
Barrier functions (or safety functions) have been established to certify set invariance in dynamical systems, while the theory of control barrier functions (CBFs) enables safe controller synthesis in control systems.
The framework of CBFs was first introduced in \cite{AmesCDC2014} and later refined in \cite{AmesXuGriTab2017}.
A comprehensive review of safety-critical control can be found in \cite{ames2019control} and the references therein.

While most works in safety-critical control are applied to delay-free systems, time delays often arise in many applications. 
For example,
human-machine interactions involve the reflex delay of the human operators,
models of vehicular traffic contain the reaction time of the drivers \cite{Ji2021}, 
wheel-shimmy motion can occur on vehicles due to the elastic contact between the tire and the road that can be modeled as distributed delay \cite{takacs2009delay},
manufacturing processes including metal cutting may suffer from vibrations due to a delayed regenerative effect of the chip formation \cite{MunoaCIRPKeynote}, 
hydraulic systems showcase time delay caused by wave propagation in pipes \cite{kadar2021time},
and epidemiological models contain delays due to the incubation period of infectious diseases \cite{rost2008seir,casella2021covid19}.
Time delay also plays important role in population dynamics \cite{kuang1993delay}, neural networks \cite{orosz2010controlling}, brain dynamics \cite{stepan2009delay_brain}, the human sensory system \cite{stepan2009delay,insperger2013acceleration} and robotic systems \cite{stepan2001vibrations}. 
Generally speaking, delays can enter a control system in two different ways: in the control input or in the state. In both cases, time delays may render the system unsafe if controllers are designed without considering the delay. 

When delay appears in the control input the dynamics is often formulated as:
\begin{equation}
\dot{x}(t) = f \big( x(t) \big) + g \big( x(t) \big) u(t-\tau).
\end{equation}
In the case of input delay, the input affects the system after a delay period, hence it needs to be taken into account what will happen to the system in the future before the input becomes effective.
Therefore, the idea of predictor feedback \cite{Krsticbook2008,BekKrs2013,Karafyllis2017,michiels2007stability} is often used to eliminate the effect of the delay by predicting the future state from the actual state and the input history.
Comprehensive literature about safety-critical control with input delay for various applications can be found in \cite{necmiye2020, singletary2020control} for discrete time, in \cite{Jankovic2018, ames2020safety, Molnar2021lcss, Molnar2021tcst, molnar2022issf} for continuous time, while the works in \cite{abel2020constrained, Abel2021} include time-varying and multiple input delays. 

When delay appears in the state, a typical example for this class of systems is:
\begin{equation}\label{eq:statedelay_example}
\dot{x}(t) = f \big( x(t), x(t-\tau) \big) + g \big( x(t), x(t-\tau) \big) u(t).
\end{equation}
In the case of state delay, safety-critical control has not yet been fully addressed to the best of our knowledge.
Therefore, this paper is intended to tackle this problem.
We seek to find controllers for time delay systems like \eqref{eq:statedelay_example} such that safety is maintained.
Specifically, we seek to keep a certain scalar safety measure $h$ positive, wherein the safety condition may also depend on delayed states (as it was first proposed in \cite{prajna2005methods}).
For example, we require the following to hold for the system to be considered safe:
\begin{equation}\label{eq:safetycondition_example}
h \big( x(t), x(t-\tau) \big) \geq 0, \quad \forall t \geq 0.
\end{equation}

The main challenge of controlling systems with state delay originates from the infinite dimensional nature of delayed dynamics.
Namely, the state of the system is a function over the delay period, that implies an infinite dimensional state space similar to partial differential equations (PDEs).
As such, time delay systems are often described by functional differential equations (FDEs) \cite{krasovskii1963stability,hale1977theory,Stepan1989,kolmanovskii2013introduction,diekmann2012delay}.
Since the theory of FDEs relies on similar concepts to that of ordinary differential equations (ODEs), they are also viewed as ``abstract ODEs''  \cite{breda2014stability}, which describe the evolution of the state in the infinite dimensional state space.
Still, the mathematical treatment of FDEs requires special care, especially for ensuring formal safety guarantees.

Since the state of time delay systems is given by a function over the time history, scalar safety measures can be constructed as functionals of the state.
There exist a few instances of using functionals in the literature in the context of safety.
Safety verification for PDEs using barrier functionals can be found in \cite{ahmadi2017safety} and state-constrained control considering integral barrier Lyapunov functionals are discussed in \cite{tee2012control,liu2016adaptive}.
For autonomous time delay systems without control, \cite{OroAme2019} introduced the concept of safety functionals, which has been investigated further in \cite{Kiss2021} by means of discretization.
These works, however, do not address control systems with time delay.

% Razumikhin and Krasovskii Approaches, control barrier-Krasovskii functionals

%%%%%%%%%%%%%%%%%%%%%%%%%%%%%%%%%%%%%%%%%%%%%%%%%%%%%%%%%%%%%%%%%%%%%%%%%%%%%%%%%%%%%%%%%%%%
\subsection{Contributions}
%%%%%%%%%%%%%%%%%%%%%%%%%%%%%%%%%%%%%%%%%%%%%%%%%%%%%%%%%%%%%%%%%%%%%%%%%%%%%%%%%%%%%%%%%%%%

The main contribution of this research is a theoretical framework that allows control synthesis with formal safety guarantees in control systems with state delay.
While this includes, for example, systems of the form \eqref{eq:statedelay_example} with safety requirement \eqref{eq:safetycondition_example}, we discuss a much wider class of time delay systems and safety conditions.
Specifically, we build on the notions of safety functionals \cite{OroAme2019} and control barrier functions \cite{AmesXuGriTab2017} to introduce {\em control barrier functionals} as tools for safety-critical controller synthesis.
We use the theory of retarded and neutral functional differential equations to prove the underlying formal safety guarantees.

We remark that a few recent papers have also approached this problem parallel to our work.
Namely, \cite{liu2021safety} considers delays with disturbances while \cite{ren2021razumikhin,ren2022razumikhin} investigate the combination of stability and safety by the application of Razumikhin- and Krasovskii-type control Lyapunov and control barrier functionals.
% Here, we do not address the issue of stability, since it can be achieved in many other ways. 
Although these recent works share some of the ideas presented in this paper, we establish a comprehensive in-depth study that is not covered by previous works, including a wider class of control barrier functionals, an exhaustive discussion on how to calculate the derivatives of these functionals, the safety of neutral FDEs, the notion of relative degree for delay systems, and multiple application examples.
Meanwhile, we do not address questions related to stability or disturbances. 

%%%%%%%%%%%%%%%%%%%%%%%%%%%%%%%%%%%%%%%%%%%%%%%%%%%%%%%%%%%%%%%%%%%%%%%%%%%%%%%%%%%%%%%%%%%%
% \subsection{State of the Art}
%%%%%%%%%%%%%%%%%%%%%%%%%%%%%%%%%%%%%%%%%%%%%%%%%%%%%%%%%%%%%%%%%%%%%%%%%%%%%%%%%%%%%%%%%%%%

The rest of the paper is organized as follows. % The problem is formally stated in Section~\ref{sec:problem_statement}.
In Section~\ref{sec:delay_free}, safety is revisited for delay-free dynamical and control systems through the notions of safety functions and control barrier functions, respectively.
% Subsections~\ref{sec:delay_free_BF} and \ref{sec:delay_free_CBF} recall the safety of dynamical systems and control design theories, respectively.
Then, Sections~\ref{sec:time_delay_sys} and~\ref{sec:safety_crit_cont} present the major contributions of this work: formal guarantees of safety for time delay systems.
Section~\ref{sec:time_delay_sys} establishes the theoretical foundations of safety functionals that certify the safety of autonomous delayed dynamical systems, while Section~\ref{sec:safety_crit_cont} discusses safety-critical control with delay by means of control barrier functionals.
In Section~\ref{sec:examp}, we demonstrate safety-critical control on illustrative examples, while Section~\ref{sec:predatorprey} presents a more practical case study through the regulated delayed predator-prey problem.
Finally, we conclude our results and discuss future research directions in Section \ref{sec:conclusion}.

\section{Safety of delay-free systems}\label{sec:delay_free}
%%%%%%%%%%%%%%%%%%%%%%%%%%%%%%%%%%%%%%%%%%%%%%%%%%%%%%%%%%%%%%%%%%%%%%%%%%%%%%%%%%%%%%%%%%%%

In this section, we revisit safety certification for delay-free dynamical systems and safety-critical control for delay-free control systems, that are described by ordinary differential equations (ODEs).
Specifically, we focus on the notions of {\em safety functions} and {\em control barrier functions}.
% that can be used to ensure forward invariance of a given domain in state space.
Then, in Section~\ref{sec:time_delay_sys}, we extend these frameworks to time delay systems.

%%%%%%%%%%%%%%%%%%%%%%%%%%%%%%%%%%%%%%%%%%%%%%%%%%%%%%%%%%%%%%%%%%%%%%%%%%%%%%%%%%%%%%%%%%%%
\subsection{Dynamical Systems}\label{sec:delay_free_BF}
%%%%%%%%%%%%%%%%%%%%%%%%%%%%%%%%%%%%%%%%%%%%%%%%%%%%%%%%%%%%%%%%%%%%%%%%%%%%%%%%%%%%%%%%%%%%

Consider the dynamical system described by the ODE:
\begin{equation}\label{eq:ODE}
\dot{x}(t) = f\big(x(t)\big),
\end{equation}
where dot represents derivative with respect to time $t$, ${x \in \R^n}$ is the state variable, and ${f \colon \R^n \to \R^n}$ is a locally Lipschitz continuous function.
Given an initial condition ${x(0) \in \R^n}$, this system has a unique solution over ${t \in I(x(0))}$ with an interval of existence $I(x(0)) \subseteq \R$.
For simplicity of exposition, throughout this paper we assume ${I(x(0)) = \R_{\geq 0}}$, i.e., solutions exist ${\forall t \geq 0}$.

We consider system \eqref{eq:ODE} safe when the solution $x(t)$ evolves within a {\em safe set} ${S \subset \R^n}$, as given by the following definition.
\begin{definition}[\textbf{Safety and Forward Invariance}]
% \textit{
System \eqref{eq:ODE} is safe w.r.t. set ${S \subset \R^n}$, if $S$ is forward invariant w.r.t. \eqref{eq:ODE} such that ${x(0) \in S \implies x(t) \in S}$, ${\forall t \geq 0}$ for the solution of \eqref{eq:ODE}.
% }
\end{definition}
Specifically, we consider set $S$ to be the 0-superlevel set of a continuously differentiable function ${h \colon \R^n \to \R}$, that is:
\begin{equation}\label{eq:safesetcond}
S   = \{x \in \R^n : h(x) \geq 0 \}.
\end{equation}
In this case, safety functions allow us to certify the safety of~\eqref{eq:ODE}.

\begin{definition}[\textbf{Safety Function}]\label{def:SF_ODE}
A continuously differentiable function ${h \colon \R^n \to \R}$ is a \textbf{safety function} for \eqref{eq:ODE} on $S$ defined by~\eqref{eq:safesetcond} if there exists ${\alpha \in \Keinf}$ (see footnote\footnote{Function ${\alpha \colon \R \to \R}$ is of extended class-$\Kinf$, denoted as ${\alpha \in \Keinf}$, if $\alpha$ is continuous, monotonically increasing, ${\alpha(0)=0}$ and ${\lim_{r \to \pm \infty} \alpha(r) = \pm \infty}$.})
such that ${\forall x \in \R^n}$:
\begin{equation}\label{eq:safecond}
\dot{h}(x) \geq -\alpha\big(h(x)\big),
\end{equation}
where ${\dot{h}(x) = \nabla h(x)  f(x) = L_f h(x)}$ is the derivative of $h$ along \eqref{eq:ODE} that is equal to the Lie derivative $L_f h$ of $h$ along $f$. Here ${\nabla h(x) \in \mathbb{R}^{1 \times n} }$ is a row vector while ${f(x) \in \mathbb{R}^{n \times 1} }$ is a column vector, and ${\nabla h(x)  f(x)}$ denotes the scalar product of these vectors. 
% }
\end{definition}

Further technical details with discussion about $\alpha$ can be found in \cite{Konda2021}.
With this definition, the main result of~\cite{AmesXuGriTab2017} establishes the safety of dynamical systems.

\begin{theorem}\cite{AmesXuGriTab2017}\label{thm:safeODE}
\textit{
Set $S$ in \eqref{eq:safesetcond} is forward invariant w.r.t. \eqref{eq:ODE} if $h$ is a safety function for \eqref{eq:ODE} on $S$, i.e., \eqref{eq:safecond} is satisfied.
}
\end{theorem}

\begin{proof}[Proof of Theorem~\ref{thm:safeODE}]
The proof is given by the comparison lemma \cite{Khalil2002}, and for further technical nuances, please refer to \cite{Konda2021}.
To set up the comparison lemma, consider the scalar initial value problem (with ${y \in \R}$):
\begin{equation}\label{eq:proof08}
\dot{y}(t) = -\alpha \big(y(t)\big), \qquad y(0) = h(x(0)),
\end{equation}
with the solution (note that~\eqref{eq:proof08} has a unique solution because $\alpha$ is an extended class-$\Kinf$ function \cite{Konda2021}):
\begin{equation}\label{eq:proof09}
y(t) = \beta \big(h(x(0)),t\big),
\end{equation}
for ${t \geq 0}$, where ${\beta \in \Keinf \mathcal{L}}$ (see footnote\footnote{Function ${\beta \colon \R \times \R_{\geq 0} \to \R}$ is of class-${\Keinf \mathcal{L}}$, denoted as ${\beta \in \Keinf \mathcal{L}}$, if ${\beta(.,s) \in \Keinf}$ for any ${s \in \R_{\geq 0}}$, and ${|\beta(r,.)}|$ is decreasing and ${\lim_{s \to \infty} \beta(r,s) = 0}$ for any ${r \in R}$.}).
By the comparison lemma for \eqref{eq:safecond} and \eqref{eq:proof08}, we obtain:
\begin{equation}\label{eq:proof010}
h(x(t)) \geq \beta \big(h(x(0)),t\big), \quad \forall t \geq 0.
\end{equation}
Therefore, $h(x(0)) \geq 0 \implies h(x(t)) \geq 0$, $\forall t \geq 0$, that is, $S$ is forward invariant. This completes the proof.
\end{proof}

Safety functions and Theorem~\ref{thm:safeODE} provide a useful tool for certifying the safety of dynamical systems: one needs to verify that~\eqref{eq:safecond} holds.
A similar concept can be used in control systems to design controllers that enforce the safety of the closed-loop dynamics, which is discussed next.

%%%%%%%%%%%%%%%%%%%%%%%%%%%%%%%%%%%%%%%%%%%%%%%%%%%%%%%%%%%%%%%%%%%%%%%%%%%%%%%%%%%%%%%%%%%%
\subsection{Safety-critical Control}\label{sec:delay_free_CBF}
%%%%%%%%%%%%%%%%%%%%%%%%%%%%%%%%%%%%%%%%%%%%%%%%%%%%%%%%%%%%%%%%%%%%%%%%%%%%%%%%%%%%%%%%%%%%

Now consider the affine control system:
\begin{equation}\label{eq:affine_ODE}
\dot{x}(t) = f\big(x(t)\big) + g\big(x(t)\big)u(t),
\end{equation}
with state $x \in \R^n$, control input $u \in \R^m$, and locally Lipschitz continuous functions $f \colon \R^n \to \R^n$ and $g \colon \R^{n \times m} \to \R^n$.
% function defined on an open and connected set % within $\R^n$
We seek to find a locally Lipschitz continuous controller $k \colon \R^n \to \R^m$, ${u=k(x)}$ to enforce that the closed-loop system:
\begin{equation}\label{eq:closedloop}
\dot{x}(t) = f\big(x(t)\big) + g\big(x(t)\big) k\big(x(t)\big),
\end{equation}
is safe w.r.t. the set $S$ in~\eqref{eq:safesetcond}, that is, $S$ is forward invariant w.r.t.~\eqref{eq:closedloop}.
Note that the local Lipschitz continuity of $f$, $g$ and $k$ ensures that~\eqref{eq:closedloop} has a unique solution $x(t)$ for any initial condition $x(0) \in \R^n$, and for simplicity we assume that the interval of existence is $t \geq 0$.
The safety requirement motivates the following definition.

\begin{definition}[\textbf{Control Barrier Function, CBF~\cite{AmesXuGriTab2017}}]\label{def:CBF}
% \textit{
A continuously differentiable function ${h \colon \R^n \to \R}$ is a \textbf{control barrier function} (CBF) (see footnote\footnote{In the literature, the term control barrier function often used for functions that go to infinity at the boundary of the safe set, whereas functions that are zero at the boundary shall be called control safety functions.
However, these terminologies are often used interchangeably in the literature, hence we use the more popular term, CBF.})
for \eqref{eq:affine_ODE} on $S$ defined by~\eqref{eq:safesetcond},
if there exists $\alpha \in \Keinf$ such that $\forall x \in \R^n$:
\begin{equation}\label{eq:CBF_condition}
\sup_{u\in \R^m} \big[\dot{h}(x,u) \big] > - \alpha\big( h(x) \big),
\end{equation}
where
\begin{equation}\label{eq:hdot}
\dot{h}(x,u)= \underbrace{\nabla h(x)  f(x)}_{L_f h(x)} + \underbrace{\nabla h(x)  g(x)}_{L _g  h(x)} \, u
\end{equation}
is the derivative of $h$ along \eqref{eq:affine_ODE}, given by the Lie derivatives $L_f h$ and $L_g h$ of $h$ along $f$ and $g$.
\end{definition}
With this definition, control systems can be rendered safe with the following extension of Theorem~\ref{thm:safeODE}.
\begin{theorem}\cite{AmesXuGriTab2017}\label{thm:safety}
\textit{
If $h$ is a CBF for \eqref{eq:affine_ODE} on $S$ defined by~\eqref{eq:safesetcond}, then any locally Lipschitz continuous controller $k \colon \R^n \to \R^m$, $u=k(x)$ satisfying:
\begin{equation}\label{eq:safety_cond_control}
\dot{h}\big(x,k(x)\big) \geq - \alpha\big(h(x)\big)
\end{equation}
${\forall x \in S}$ renders set $S$ forward invariant w.r.t.~\eqref{eq:closedloop}.
}
\end{theorem}

\begin{proof}[Proof\cite{AmesXuGriTab2017}]
Definition~\ref{def:CBF} ensures that controller $k$ exists. Then, considering the closed-loop dynamics in~\eqref{eq:closedloop}, the set $S$ is forward invariant according to Theorem~\ref{thm:safeODE}.
\end{proof}

This result provides systematic means to safety-critical controller synthesis: one needs to satisfy condition~\eqref{eq:safety_cond_control} when designing the controller.
Condition~\eqref{eq:safety_cond_control} can be incorporated into optimization problems as constraint to find pointwise optimal safety-critical controllers.
For example, one can modify a desired controller $k_{\rm des}$ in a minimally invasive fashion to a safe controller $k$ by solving a quadratic program, as stated formally below.
\begin{corollary}
\textit{
Given a CBF $h$ and a locally Lipschitz continuous desired controller ${k_{\rm des}\colon \R^n \to \R^m}$, $u_{\rm des}=k_{\rm des}(x)$, the following quadratic program (QP) yields a controller ${k\colon \R^n \to \R^m}$, $u=k(x)$ that renders set $S$ in~\eqref{eq:safesetcond}  forward invariant w.r.t.~\eqref{eq:closedloop}:
\begin{equation}\label{eq:qp_min_norm}
\begin{split}
k(x) =
\underset{u \in \R^m}{\operatorname{argmin}}
% \arg\min_{ \hspace{-.5cm} u \in \R^m}
& \quad \frac{1}{2} \|u - k_{\rm des}(x)\|_2^2  \\
\mathrm{s.t.} & \quad \dot{h}(x,u) \geq - \alpha\big(h(x)\big).
\end{split}
\end{equation}
Furthermore, the explicit solution of~\eqref{eq:qp_min_norm} can be found by the Karush–Kuhn–Tucker (KKT) \cite{Boyd2004} conditions as \cite{Molnar2021tcst}:
\begin{equation}
k(x)= 
\begin{cases}
k_{\rm des}(x) & {\rm if} \; \varphi(x) \geq 0,\\
k_{\rm des}(x)-\frac{\varphi(x) \varphi_0^\top(x)}{\varphi_0 (x)\varphi_0^\top(x)}  & {\rm otherwise},
\end{cases}
\end{equation}
where $\varphi(x) =  L_f h(x) + L_g h(x)\, k_{\rm des}(x) + \alpha \big( h(x) \big)$
and $\varphi_0 (x)= L_g h(x)$.
}
\end{corollary}

We remark that, with some extra care, the CBF condition~\eqref{eq:CBF_condition} could be prescribed for ${\forall x \in \R^n}$ (rather than for ${\forall x \in S}$ only). This results in the attractivity of $S$, i.e., CBFs render $S$ stabilizable, allowing the system to return to safe states from unsafe ones~\cite{xu2015robustness}.
Furthermore,~\eqref{eq:CBF_condition} is equivalent to:
\begin{equation}\label{eq:CBF_condition_Lgh}
L_g h(x) = 0 \implies L_f h(x) > -\alpha\big(h(x)\big).
\end{equation}
When ${L_g h(x) = 0}$, the derivative of $h$ with respect to time is not affected by the control input $u$ according to~\eqref{eq:hdot}, hence the corresponding uncontrolled system must be safe on its own.
One may sufficiently satisfy~\eqref{eq:CBF_condition_Lgh} by requiring that $L_g h(x)$ is never zero.
Then, the first derivative of $h$ is always affected by the control input $u$, which is referred to as $h$ has {\em relative degree}~$1$.
Otherwise, higher derivatives of $h$ may be affected by $u$ that leads to higher relative degrees, given by the following definition.
\begin{definition}[\textbf{Relative Degree}\cite{Isidori_NCS}]\label{def:reldeg}
Function ${h \colon \R^n \to \R}$ has \textbf{relative degree} $r$ (where ${r \in \mathbb{Z}}$, ${r \geq 1}$) w.r.t.~\eqref{eq:affine_ODE} if it is $r$ times continuously differentiable and the following holds $\forall x \in \R^n$:
\begin{align}
\begin{split}
    L_g L_f^{r-1} h(x) & \neq 0, \\
    L_g L_f^k h(x) & = 0, \quad {\rm for} \; r \geq 2, \; k \in \{0, \ldots, r-2 \}, \\
\end{split}
\end{align}
where $L_g L_f^0 h(x) = L_g h(x)$ and the second condition only applies for $r \geq 2$.
\end{definition}
\noindent For higher relative degrees, there exist systematic methods to construct CBFs from $h$ and guarantee safety; see~\cite{Nguyen2016, Xiao2019, sarkar2020highrelative, Wang2020LearningCB} for details.
We will show that the relative degree may be significantly affected when time delay is included in the system \cite{germani2001asymptotic}.

% Note that the first derivative of $h$ with respect to time is affected by the control input $u$ through the term $L_g h(x)$; cf.~\eqref{eq:hdot}.
% When ${L_g h(x) \neq 0}$, ${\forall x \in S}$, we say $h$ has {\em relative degree}~$1$.
% % Note that for $\phi_{1}^+(x)$ to exist, one needs ${L_g h \neq 0}$, ${\forall x \in S}$.
% % This is often referred to as $h$ has {\em relative degree}~$1$ (i.e., the first derivative of $h$ with respect to time is affected by the input $u$).
% For higher relative degrees (i.e., when higher derivatives of $h$ are affected by $u$), there exist systematic methods to construct CBFs from $h$ and guarantee safety; see~\cite{Nguyen2016, Xiao2019, sarkar2020highrelative, Wang2020LearningCB} for details.
% As shown below, the relative degrees may change significantly when time delay is included in the system.

%%%%%%%%%%%%%%%%%%%%%%%%%%%%%%%%%%%%%%%%%%%%%%%%%%%%%%%%%%%%%%%%%%%%%%%%%%%%%%%%%%%%%%%%%%%%
\section{Safety of time delay systems}\label{sec:time_delay_sys}
%%%%%%%%%%%%%%%%%%%%%%%%%%%%%%%%%%%%%%%%%%%%%%%%%%%%%%%%%%%%%%%%%%%%%%%%%%%%%%%%%%%%%%%%%%%%
In what follows, we extend the above concepts to certify safety and provide safety-critical controllers for time delay systems.

We denote the delay as $\tau > 0$, and we will rely on the following function spaces~\cite{Kim1999}.
\begin{itemize}
\item $\mathcal{C}$: the space of continuous functions mapping from $[-\tau,0]$ to $\R^n$ with norm $\| \phi \| = \max_{\vartheta \in [-\tau,0]} \| \phi(\vartheta) \|_2$ for $\phi \in \mathcal{C}$.
\item $\Q$: the space of functions mapping from $[-\tau,0]$ to $\R^n$ that are continuous almost everywhere on $[-\tau,0]$ except possibly for a finite set of points with discontinuity of the first kind.
The corresponding norm is $\| \phi \| = \sup_{\vartheta \in [-\tau,0]} \| \phi(\vartheta) \|_2$ for $\phi \in \mathcal{Q}$. \item $\B$: the space of continuous functions with almost everywhere continuous derivatives, i.e., ${\B = \{ \phi \in \mathcal{C}: \dot{\phi} \in \Q \}}$, equipped with the norm $\| \phi \| = \max_{\vartheta \in [-\tau,0]} \| \phi(\vartheta) \|_2$ for $\phi \in \B$.
Throughout the paper we use the convention that the derivative $\dot{\phi}$ indicates right-hand derivative at the discontinuity points.
\end{itemize}

%%%%%%%%%%%%%%%%%%%%%%%%%%%%%%%%%%%%%%%%%%%%%%%%%%%%%%%%%%%%%%%%%%%%%%%%%%%%%%%%%%%%%%%%%%%%
\subsection{Retarded Dynamical Systems}\label{sec:RFDE}
%%%%%%%%%%%%%%%%%%%%%%%%%%%%%%%%%%%%%%%%%%%%%%%%%%%%%%%%%%%%%%%%%%%%%%%%%%%%%%%%%%%%%%%%%%%%

First, we consider dynamical systems with time delays where the rate of change of the state depends on the past values of the state.
This class of systems is referred as retarded functional differential equation (RFDE) and is given in the form:
\begin{equation}\label{eq:Aut_FDE}
    \dot{x}(t) = \F(x_t),
\end{equation}
where $x \in \R^n$ is the state variable and $x_t: [-\tau,0] \to \R^n$ represents the history of the state over $[t-\tau,t]$ with delay $\tau > 0$. According to the standard Hale-Krasovsky notation \cite{krasovskii1963stability, hale1977theory}, it is defined by the shift:
\begin{equation}\label{eq:state_function}
x_t(\vartheta) =  x(t+\vartheta), \; \vartheta \in [-\tau,0],
\end{equation}
that is an element of the Banach space $\B$ defined above.
Functional $\F:\B \to \R^n$ is assumed to be locally Lipschitz continuous, hence~\eqref{eq:Aut_FDE} has a unique solution over a time interval ${t \in I(x_0)}$ for any initial state history $x_0 \in \B$ \cite{HaleLunel1993}.
Again, for simplicity of exposition, we assume ${I(x_{0}) = \R_{\geq 0}}$.
We define the derivative of the solution by $\dot{x}_t \in \Q$:
% $\dot{x}_t \colon [-\tau,0] \to \R^n$, 
\begin{equation}\label{eq:state_derivative}
\dot{x}_t(\vartheta)=
\begin{cases}
\F(x_t) & {\rm if} \; \vartheta = 0,\\
\dot{x}(t+\vartheta) & {\rm if} \; \vartheta \in[-\tau, 0).
\end{cases}
\end{equation}

Since the state space of~\eqref{eq:Aut_FDE} is $\B$, which is infinite dimensional, one needs the state $x_t$ to evolve within an infinite dimensional {\em safe set} $\mathcal{S} \subset \B$ to certify safety as given by the following definition.

\begin{definition}[\textbf{Safety and Forward Invariance of RFDE}]\label{def:safety}
System \eqref{eq:Aut_FDE} is safe w.r.t. set ${\mathcal{S} \subset \B}$, if $\mathcal{S}$ is forward invariant w.r.t. \eqref{eq:Aut_FDE} such that ${x_0 \in \mathcal{S} \implies x_t \in \mathcal{S}}$, ${\forall t \geq 0}$ for the solution of \eqref{eq:Aut_FDE}.
\end{definition}

Assume that the set $\mathcal{S}$ can be constructed as the 0-superlevel set of a continuously Fr\'{e}chet differentiable functional ${\mathcal{H} \colon \B \to \R}$ (rather than a function):
\begin{equation}\label{eq:safesetcond_delay}
    \mathcal{S} = \{x_t \in \B : \mathcal{H} (x_t) \geq 0 \}.
\end{equation}
This makes it evident that safety functions defined over the finite dimensional space $\R^n$ are not adequate to establish safety.
Instead, one may construct so-called {\em safety functionals} defined over the infinite dimensional state space $\B$. 
\begin{definition}[\textbf{Safety Functional~\cite{OroAme2019}}]\label{def:SF_DDE}
A continuously Fr\'{e}chet differentiable functional ${\mathcal{H} \colon \B \to \R}$ is a \textbf{safety functional} for \eqref{eq:Aut_FDE} on $\mathcal{S}$ defined by~\eqref{eq:safesetcond_delay} if there exists ${\alpha \in \Keinf}$ such that ${\forall x_t \in \B}$:
\begin{equation}\label{eq:safecondDDE}
\dot{\mathcal{H}}(x_t,\dot{x}_t) \geq -\alpha \big(\mathcal{H}(x_t)\big),
\end{equation}
where ${\dot{\mathcal{H}} := \mathcal{L}_\F \mathcal{H} \colon \B \times \Q \to \R}$ is the derivative of $\mathcal{H}$ along \eqref{eq:Aut_FDE} with $\dot{x}_t$ given in~\eqref{eq:state_derivative}.
\end{definition}

Note that the derivative of $\mathcal{H}$ depends both on $x_t$ and $\dot{x}_t$ in~\eqref{eq:state_derivative}, and its calculation is detailed below.
The following theorem from~\cite{OroAme2019} certifies safety, i.e., the forward invariance of $\mathcal{S}$ via safety functionals.

\begin{theorem}\cite{OroAme2019}\label{thm:safe_RFDE}
\textit{
Set $\mathcal{S}$ in \eqref{eq:safesetcond_delay} is forward invariant w.r.t. \eqref{eq:Aut_FDE} if $\mathcal{H}$ is a safety functional for \eqref{eq:Aut_FDE} on $\mathcal{S}$, i.e., \eqref{eq:safecondDDE} is satisfied.
% Given the set $\mathcal{S} \subset \B$ that is a super level set of a continuously differentiable functional $\mathcal{H} \colon \B \to \R$, it is forward invariant if
% \begin{equation}\label{eq:safecondDDE}
% \dot{\mathcal{H}}(x_t,\dot{x}_t) \geq -\alpha(\mathcal{H}(x_t)),
% \end{equation}
% where ${\dot{\mathcal{H}} := \mathcal{L}_\F \mathcal{H} \colon \B \times \B \to \R }$ is the directional derivative along solutions of \eqref{eq:Aut_FDE} and $\alpha$ is a class $\K$ function. In this case $\mathcal{S}$ is called the \textbf{safe set} while $\mathcal{H}$ is called the \textbf{safety functional} or \textbf{barrier functional}.
}
\end{theorem}

\begin{proof}[Proof\cite{OroAme2019}]
The proof follows from the comparison lemma.
We set up the initial value problem (with $y \in \R$):
\begin{equation}\label{eq:proof8}
\dot{y}(t) = -\alpha \big(y(t)\big), \qquad y(0) = \mathcal{H}(x_0),
\end{equation}
which has the solution (which is unique, as in the case of the proof of Theorem~\ref{thm:safety}, because $\alpha$ is an extended class-$\Kinf$ function):
\begin{equation}\label{eq:proof9}
y(t) = \beta \big(\mathcal{H}(x_0),t\big)
\end{equation}
for ${t \geq 0}$, where
${\beta \in \Keinf \mathcal{L}}$.
Then, applying the comparison lemma for \eqref{eq:safecondDDE} and \eqref{eq:proof8}, we obtain:
\begin{equation}\label{eq:proof10}
\mathcal{H}(x_t) \geq \beta \big(\mathcal{H}(x_0),t\big), \quad \forall t \geq 0.
\end{equation}
This leads to ${\mathcal{H}(x_0) \geq 0 \implies \mathcal{H}(x_t) \geq 0}$, ${\forall t \geq 0}$, that is, $\mathcal{S}$ is forward invariant w.r.t.~\eqref{eq:Aut_FDE}.
This completes the proof. 
\end{proof}

%%%%%%%%%%%%%%%%%%%%%%%%%%%%%%%%%%%%%%%%%%%%%%%%%%%%%%%%%%%%%%%%%%%%%%%%%%%%%%%%%%%%%%%%%%%%
\subsection{Time Derivative of Safety Functional} %$\mathcal{H}$
%%%%%%%%%%%%%%%%%%%%%%%%%%%%%%%%%%%%%%%%%%%%%%%%%%%%%%%%%%%%%%%%%%%%%%%%%%%%%%%%%%%%%%%%%%%%

The left-hand side of~\eqref{eq:safecondDDE} is the time derivative of $ \mathcal{H}$
along the solution of~\eqref{eq:Aut_FDE}.
While for finite dimensional delay-free systems the derivative $\dot{h}$ is given by the directional derivative or Lie derivative~\cite{krasovskii1963stability,oguchi2002input},
the derivative $\dot{\mathcal{H}}$ in the presence of time delay and infinite dimensional dynamics has an intricate representation which we break down below.

Recall that in the delay-free case, the derivative $\dot{h}$ of the safety function is calculated by the chain rule as ${\dot{h}(x) = \nabla h(x)  \dot{x}}$ (where $\dot{x} = f(x)$).
That is, $\dot{h}$ is given by a linear function of $\dot{x}$.
Similarly, for time delay systems the derivative $\dot{\mathcal{H}}$ of the safety functional can be given by a linear functional of the state derivative $\dot{x}_t$ in~\eqref{eq:state_derivative}.
This is stated by the theorem below.
\begin{theorem}\label{thm:H_func_deriv}
\textit{
Consider system~\eqref{eq:Aut_FDE} and let $\mathcal{H}\colon\B\to\R$ be a continuously Fr\'{e}chet differentiable functional.
Then there exists a unique ${\eta \colon \B \times [-\tau,0] \to \R^{1 \times n}}$ that is of bounded variation~\cite{HaleLunel1993, diekmann2012delay} in its second argument such that the time derivative of $\mathcal{H}$ along~\eqref{eq:Aut_FDE} can be expressed as:
\begin{equation}\label{eq:thm_bounded_var}
\dot{\mathcal{H}}(x_t,\dot{x}_t) = \int_{-\tau}^{0} {\rm d}_\vartheta \eta(x_t,\vartheta) \dot{x}_t(\vartheta),
\end{equation}
where the integral is a Stieltjes type.
% and $\dot{x}_t$ is given in~\eqref{eq:state_derivative}.
% \begin{equation}
% \dot{x}_t(\vartheta)=
% \begin{cases}
% \F\big(x_t(\vartheta)\big)     & {\rm if} \; \vartheta = 0,\\
% \dot{x}(t+\vartheta) & {\rm if} \; \vartheta \in[-\tau, 0).
% \end{cases}
% \end{equation}
}
\end{theorem}

\begin{proof}
% [Proof of Theorem~\ref{thm:H_func_deriv}]
Here we provide the main steps of the proof and the remaining details (definitions and a lemma) are in Appendix~\ref{sec:appdx_Gateaux}.
% Using ${\mathcal{H}(x_{t+\Delta t})=\mathcal{H}(x_t + \Delta t \dot{x}_t}$,
The derivative of $\mathcal{H}$ along~\eqref{eq:Aut_FDE} can be expressed by the G\^{a}teaux derivative (see definition at~\eqref{eq:gateaux_definition}) along $\dot{x}_t$ as:
\begin{equation}\label{eq:proof_1}
  \dot{\mathcal{H}}(x_t,\dot{x}_t) = 
  \lim_{\Delta t \to 0} \frac{\mathcal{H}(x_{t+\Delta t}) - \mathcal{H}(x_t) }{\Delta t} 
=  \lim_{\Delta t \to 0} \frac{\mathcal{H}(x_t + \Delta t \dot{x}_t)-\mathcal{H}(x_t) }{\Delta t} 
= D_{\rm G}\mathcal{H}(x_t)(\dot{x}_t),
\end{equation}
where
% $\dot{x}_t$ is in~\eqref{eq:state_derivative} and
$D_{\rm G}\mathcal{H}(x_t) \colon \Q \to \R$ denotes the G\^{a}teaux derivative at $x_t$.
Since $\mathcal{H}$ is continuously Fr\'{e}chet differentiable, the G\^{a}teaux derivative exists and can be given by the Fr\'{e}chet derivative (see definition at~\eqref{eq:frechet_definition}) based on Lemma~\ref{lem:gateaux_is_frechet}:
\begin{equation}
D_{\rm G}\mathcal{H}(x_t)(\dot{x}_t)=D_{\rm F}\mathcal{H}(x_t)\dot{x}_t.
\end{equation}
Here $D_{\rm F}\mathcal{H}(x_t) \colon \Q \to \R$ is the Fr\'{e}chet derivative at $x_t$ that is a continuous linear functional (applied on $\dot{x}_t$).
Then, the Riesz representation theorem (see Theorem~\ref{thm:riesz} in Appendix~\ref{sec:appdx_Gateaux}) implies that there exists a unique ${\eta \colon \B \times [-\tau,0] \to \R^{1 \times n}}$ that is of bounded variation~\cite{HaleLunel1993, diekmann2012delay} in its second argument such that this linear functional can be expressed as the Stieltjes integral:
% (inner product):
\begin{equation}
D_{\rm F}\mathcal{H}(x_t)\dot{x}_t=%\langle  \dot{x}_t,  \eta(x_t) \rangle =
\int_{-\tau}^{0} {\rm d}_\vartheta \eta(x_t,\vartheta) \dot{x}_t(\vartheta),
\end{equation}
which leads to~\eqref{eq:thm_bounded_var} and proves the statement in Theorem~\ref{thm:H_func_deriv}.
\end{proof}

To summarize, in the delay-free case function $h$ is continuously differentiable, the gradient $\nabla h(x)$ exists, and it allows the calculation of the directional derivatives of $h$ in any direction, including $\dot{x}$.
Similarly, with delay we assume that
% the derivatives of the functional $\mathcal{H}$ exists in all directions, that is, it
the functional $\mathcal{H}$ is continuously Fr\'{e}chet differentiable, and hence the G\^{a}teaux derivative can be expressed along any direction, including $\dot{x}_t$.
In this sense, the integral with $\eta$ is the infinite dimensional counterpart of the scalar product with the gradient $\nabla h(x)$.

The expression of $\eta$ depends on the specific form of $\mathcal{H}$ (just as the expression of $\nabla h(x)$ depends on $h$).
For example, when $\mathcal{H}(x_t)$ involves both delay-free states, multiple discrete (point) delays ${\tau_j \in [-\tau,0]}$, ${j \in \{1, \ldots, l\}}$, and a continuous (distributed) delay over $[-\sigma_1,-\sigma_2] \subseteq [-\tau,0]$,
the bounded variation.
It has the form:
\begin{equation}\label{eq:bounded_variation}
\eta(x_t,\vartheta) = w_0(x_t) \theta(\vartheta) + \sum_{j=1}^l w_j(x_t)\hat{\theta}(\vartheta+\tau_j) + 
\eta_{\rm d}(x_t,\vartheta)
\end{equation}
with 
\begin{equation}
\eta_{\rm d}(x_t,\vartheta) = 
\begin{cases}
0  & {\rm if} \quad \vartheta < -\sigma_1,\\
\int_{-\sigma_1}^{\vartheta} w_{\rm d}(x_t,s) \,{\rm d}s & {\rm if} \quad -\sigma_1 \leq \vartheta \leq -\sigma_2,\\
\int_{-\sigma_1}^{-\sigma_2} w_{\rm d}(x_t,s) \,{\rm d}s & {\rm if} \quad -\sigma_2 < \vartheta ,\\
\end{cases} 
\end{equation}
see Fig.~\ref{fig:NBV} for illustration.
Here the weights ${w_0,w_j \colon \mathcal{B} \to \mathbb{R}^{1 \times n}}$ and ${w_{\rm d} \colon \mathcal{B} \times \mathbb{R} \to \mathbb{R}^{1 \times n}}$ are (potentially complicated nonlinear) functionals that depend on the form of $\mathcal{H}$ (an example is given below).
Furthermore, $\theta$ and $\hat{\theta}$ denote the right and left continuous Heaviside step functions:
\begin{equation}
\theta(s) = 
\begin{cases}
1  & {\rm if} \quad s \geq 0,\\
0  & {\rm if} \quad s < 0,\\
\end{cases}, \qquad
\hat{\theta}(s)=1-\theta(-s) = 
\begin{cases}
1  & {\rm if} \quad s > 0,\\
0  & {\rm if} \quad s \leq 0.\\
\end{cases}
\end{equation}

The integral in~\eqref{eq:thm_bounded_var} can also be expressed as a distribution that comprises of finitely many shifted Dirac delta distributions (corresponding to the point delays) and a bounded kernel (corresponding to the remaining distributed delay) in the form:
\begin{equation}\label{eq:wight_function}
\dot{\mathcal{H}}(x_t,\dot{x}_t) = 
\int_{-\infty}^{\infty}  w(x_t,\vartheta) \dot{x}_t(\vartheta) \,{\rm d}\vartheta,
\end{equation}
where the kernel $w\colon\B \times \R \to \R^{1 \times n}$
reads:
\begin{equation}\label{eq:Riesz_2}
w(x_t,\vartheta) = 
w_0(x_t) \delta(\vartheta) + \sum_{j=1}^l w_j(x_t)\delta(\vartheta+\tau_j) + 
\tilde{w}_{\rm d}(x_t,\vartheta)
\end{equation}
with $\delta$ being the Dirac delta and
\begin{equation}
\tilde{w}_{\rm d}(x_t,\vartheta) = 
\begin{cases}
0  & {\rm if} \quad \vartheta < -\sigma_1,\\
\int_{-\sigma_1}^{\vartheta} w_{\rm d}(x_t,s) \,{\rm d}s  & {\rm if} \quad -\sigma_1 \leq \vartheta \leq -\sigma_2,\\
0  & {\rm if} \quad -\sigma_2 < \vartheta.\\
\end{cases} 
\end{equation}
Substitution of~\eqref{eq:Riesz_2} and~\eqref{eq:state_derivative} into~\eqref{eq:wight_function} leads to the following conclusion for the derivative of $\mathcal{H}$ along \eqref{eq:Aut_FDE}.

\begin{corollary}
\textit{
When $\mathcal{H}(x_t)$ involves both delay-free states, point delays ${\tau_j \in [-\tau,0]}$, ${j \in \{1, \ldots, l\}}$, and a distributed delay over $[-\sigma_1,-\sigma_2] \subseteq [-\tau,0]$,
the time derivative of $\mathcal{H}$ in~\eqref{eq:thm_bounded_var} along \eqref{eq:state_derivative} is of the form:
\begin{equation}\label{eq:functional_H_dot}
\begin{split}
\dot{\mathcal{H}}(x_t,\dot{x}_t)
= \mathcal{L}_\F \mathcal{H}(x_t,\dot{x}_t)
:=
w_0(x_t) 
\F(x_t) + 
\sum_{j=1}^l w_j(x_t)\dot{x}_t(-\tau_j) % \\
+ \int_{-\sigma_1}^{-\sigma_2} w_{\rm d}(x_t,\vartheta) \dot{x}_t(\vartheta) \, {\rm d}\vartheta,
\end{split}
\end{equation}
with weights ${w_0, w_j \colon \mathcal{B} \to \mathbb{R}^{1 \times n}}$ and ${w_{\rm d} \colon \mathcal{B} \times \mathbb{R} \to \mathbb{R}^{1 \times n}}$.
}
\end{corollary}

%%%%%%%%%%%%%%%%%%%%%%%%%%%%%%%%%%%%%%%%
\begin{figure}
\centerline{\includegraphics[width=.9\columnwidth]{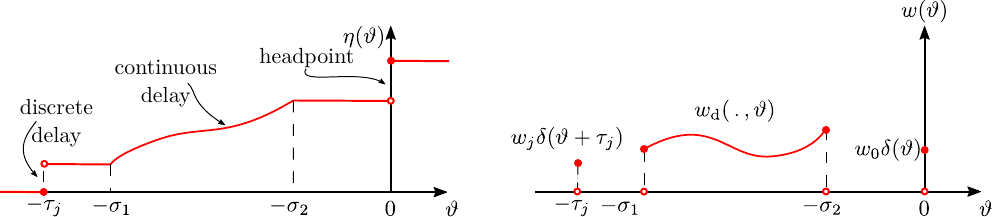}}
\caption{Illustration of the bounded variation function $\eta$ in~\eqref{eq:bounded_variation} and the distribution kernel $w$ in~\eqref{eq:Riesz_2} for a discrete delay $\tau_1$ and a continuous delay over $[-\sigma_1,-\sigma_2]$.
}
\label{fig:NBV}
\end{figure}
%%%%%%%%%%%%%%%%%%%%%%%%%%%%%%%%%%%%%%%%

This leads to two important properties about the derivative of $\mathcal{H}$, given by the following remarks.
We will rely on these properties when discussing the relative degree of control barrier functionals in Section~\ref{sec:rel_deg}.

\begin{remark}[($\mathcal{H}(x_t)$ contains $x(t)$)]\label{rem:atomic}
When $\mathcal{H}(x_t)$ contains the present state $x(t)=x_t(0)$, it is indicated by $w_0(x_t) \neq 0$, $\forall x_t \in \B$.
Then, the time derivative $\dot{\mathcal{H}}(x_t,\dot{x}_t)$ is directly affected by the right-hand side $\F(x_t)$.
This will be a necessary requirement for enforcing safety via control (i.e., when we maintain safety by designing $\F$ for a closed control loop).
\end{remark}

\begin{remark}[($\mathcal{L}_\F \mathcal{H}$ independent of $\dot{x}_t$)]\label{rem:simplification}
If $w_{\rm d}$ is continuously differentiable in $\vartheta$, with derivative denoted by $w_{\rm d}'$, then the integral $\int_{-\sigma_1}^{-\sigma_2} w_{\rm d}(x_t,\vartheta) \dot{x}_t(\vartheta)  {\rm d}\vartheta$
can be simplified.
Since $\frac{\partial}{\partial t} x_t(\vartheta)=\frac{\partial}{\partial \vartheta} x_t(\vartheta)$, integration by parts eliminates $\dot{x}_t$ and leads to an expression that depends on $x_t$ only:
\begin{equation}
\Lambda(x_t) := \int_{-\sigma_1}^{-\sigma_2} w_{\rm d}(x_t,\vartheta) \dot{x}_t(\vartheta) \, {\rm d}\vartheta = w_{\rm d}(x_t,-\sigma_2) x_t(-\sigma_2) - w_{\rm d}(x_t,-\sigma_1) x_t(-\sigma_1) - 
\int_{-\sigma_1}^{-\sigma_2} w_{\rm d}'(x_t,\vartheta) x_t(\vartheta) \, {\rm d}\vartheta.
\end{equation}
Additionally, if $w_j(x_t)=0$, $\forall x_t \in \B$, $\forall j \geq 1$ also holds, then $\mathcal{L}_\F \mathcal{H}$ can be expressed as a functional of $x_t$ only, independent of $\dot{x}_t$, which we shorty denote as $\mathcal{L}_\F \mathcal{H}(x_t)$.
Then, the time derivative of $\mathcal{L}_\F \mathcal{H}(x_t)$ can be calculated by the same method as that of $\mathcal{H}(x_t)$.
We will exploit this in higher relative degree scenarios in Section~\ref{sec:rel_deg}.
\end{remark}

We demonstrate the above calculation of the time derivative of $\mathcal{H}$ by an example below.
\begin{eexample} \label{example:functional}
Consider the system \eqref{eq:Aut_FDE} with functional $\mathcal{H}$ that contains point delays $\tau_1, \ldots, \tau_l$ and a distributed delay over $[-\sigma_1,-\sigma_2]$, defined as:
\begin{equation}\label{eq:safe_funnal_gen_exa}
\begin{split}
    \mathcal{H}(x_t) = h\bigg( x_t(0), x_t(-\tau_1),\dots, x_t(-\tau_l), \int_{-\sigma_1}^{-\sigma_2} \rho(\vartheta) \, \kappa \big(x_t(\vartheta)\big)\, {\rm d} \vartheta
    \bigg),
    \end{split}
\end{equation}
where ${h \colon \R^n \times \ldots \times \R^n \to \R}$, ${\kappa \colon \R^n \to \R^n}$ and $\rho \colon [-\sigma_1,-\sigma_2] \to \R^{n\times n}$ are continuously differentiable.
One may directly take the time derivative of \eqref{eq:safe_funnal_gen_exa} and substitute~\eqref{eq:state_derivative}, which leads to the following derivative required for certifying safety:
%, use \eqref{eq:Aut_FDE} for $\dot{x}(t)$ and considering $\frac{{\rm d}}{{\rm d} \vartheta} x_t(\vartheta)=\frac{{\rm d}}{{\rm d} t} x_t(\vartheta)$ yielding
\begin{equation}\label{eq:safe_funnal_der_exa}
\begin{split}
\dot{\mathcal{H}}(x_t,\dot{x}_t)
%= \mathcal{L}_\F  \mathcal{H}   
& =\underbrace{\nabla_0 h(\dots)}_{w_0(x_t)}  \mathcal{F}(x_t)
% \F(x_t,\dot{x}_t)
+ \sum_{j=1}^l \underbrace{\nabla_j h(\dots)}_{w_j(x_t)}  \dot{x}_t(-\tau_j) % & \\
+\int_{-\sigma_1}^{-\sigma_2}
\underbrace{\nabla_{l+1} h(\dots) \rho(\vartheta) \nabla \kappa \big(x_t(\vartheta)\big)}_{w_{\rm d}(x_t,\vartheta)} \dot{x}_t(\vartheta)\, {\rm d} \vartheta,
% \\& + \nabla n_{\rm d}  \bigg( \rho( -\sigma_2) \kappa \big( x(t-\sigma_2) \big)  
% - \rho( -\sigma_1) \kappa \big( x(t-\sigma_1) \big) 
% - \int_{-\sigma_1}^{-\sigma_2} \frac{{\rm d}\rho(\vartheta)}{{\rm d}\vartheta}\, \kappa \big( x_t(\vartheta) \big)  {\rm d} \vartheta \bigg)\\
\end{split}
\end{equation}
where $\nabla_j$ is the gradient with respect to the $j$-th argument of a function, $\nabla \kappa$ is the Jacobian of $\kappa$, and $(\dots)$ is a shorthand notation for evaluation at the argument of $h$ as in \eqref{eq:safe_funnal_gen_exa}.
\end{eexample}

While the example~\eqref{eq:safe_funnal_gen_exa} covers many practical choices of safety functionals, it does not include all possibilities for $\mathcal{H}$.
A more general example with a double integral can be found in Appendix~\ref{sec:appdx_double_int}, and one could also include triple, quadruple, etc. integrals.

%%%%%%%%%%%%%%%%%%%%%%%%%%%%%%%%%%%%%%%%%%%%%%%%%%%%%%%%%%%%%%%%%%%%%%%%%%%%%%%%%%%%%%%%%%%%
\subsection{Neutral Dynamical Systems}\label{sec:NFDE}
%%%%%%%%%%%%%%%%%%%%%%%%%%%%%%%%%%%%%%%%%%%%%%%%%%%%%%%%%%%%%%%%%%%%%%%%%%%%%%%%%%%%%%%%%%%%

Now, we address time delay systems where the rate of change of the state depends on the past values of the state as well as past state derivatives.
This class of systems is referred as neutral functional differential equation (NFDE) and is given in the form:
\begin{equation}\label{eq:NFDE}
\dot{x}(t) = \F(x_t,\dot{x}_t),
\end{equation}
where ${x \in \R^n}$ is the state variable, ${x_t \in \B}$ is the state history defined in \eqref{eq:state_function} and ${\dot{x}_t \in \Q}$ is its derivative. The functional  
${\F \colon \B \times \Q \to \R^n}$ 
is assumed to be locally Lipschitz continuous in its first argument ($x_t$), continuous in its second argument ($\dot{x}_t$), and assumed to satisfy a modified Lipschitz condition that takes the form:
\begin{equation}\label{eq:neutral_Lips}
\| \F\big(x_t,\phi_1)-\F\big(x_t,\phi_2) \|_2 \leq  L \|\phi_1-\phi_2\|, \quad L<1,
\end{equation}
for any functions ${\phi_1,\phi_2 \in \Q}$ such that ${\phi_1(\vartheta)=\phi_2(\vartheta)}$, ${\forall \vartheta \in[-\tau, -\delta]}$ for some ${\delta>0}$.

According to Theorem 3.2 in \cite{kolmanovskii1986stability}, if the condition~\eqref{eq:neutral_Lips} holds, system~\eqref{eq:NFDE} has a unique (continuous but potentially nonsmooth) solution over a time interval ${t \in I(x_0,\dot{x}_0)}$ for initial history $x_0 \in \B$ and $\dot{x}_0 \in \Q$.
This existence and uniqueness result was first proved by \cite{zverkin1962existence} and later also summarized in \cite{akhmerov1984theory,angelov1980existence}.
The derivative of the solution then becomes:
\begin{equation}\label{eq:state_derivative_neutral}
\dot{x}_t(\vartheta)=
\begin{cases}
\F(x_t,\dot{x}_t) & {\rm if} \; \vartheta = 0,\\
\dot{x}(t+\vartheta) & {\rm if} \; \vartheta \in[-\tau, 0).
\end{cases}
\end{equation}
For comprehensive details on the properties of NFDEs, please refer to~\cite{krasovskii1963stability,hale1977theory}.

The safety of~\eqref{eq:NFDE} can be formulated the same way as it was done for retarded systems in Section~\ref{sec:RFDE}.
Safety means that the state $x_t$ evolves within the safe set $\mathcal{S} \subset \B$, which is constructed by the safety functional $\mathcal{H}: \B \to \R$ (see Definitions~\ref{def:safety} and~\ref{def:SF_DDE}).
Then, the safety of the neutral system \eqref{eq:NFDE} is formally certified by the following corollary.
\begin{corollary}\label{cor:safe_NFDE}
\textit{
Set $\mathcal{S}$ in \eqref{eq:safesetcond_delay} is forward invariant w.r.t. \eqref{eq:NFDE} if $\mathcal{H}$ is a safety functional for \eqref{eq:NFDE} on $\mathcal{S}$, i.e.:
\begin{equation}\label{}
\dot{\mathcal{H}}(x_t,\dot{x}_t) \geq -\alpha\big(\mathcal{H}(x_t)\big),
\end{equation}
is satisfied, where ${\dot{\mathcal{H}} := \mathcal{L}_\F \mathcal{H} \colon \B \times \Q \to \R}$ is the derivative of $\mathcal{H}$ along \eqref{eq:NFDE} with $\dot{x}_t$ given in~\eqref{eq:state_derivative_neutral}.
}
\end{corollary}
The difference between retarded and neutral systems is that for the neutral case the derivative of $\mathcal{H}$ is provided by substituting~\eqref{eq:state_derivative_neutral} into~\eqref{eq:thm_bounded_var} (while we used \eqref{eq:state_derivative} for retarded systems).
Analogously to \eqref{eq:functional_H_dot}, we can write:
\begin{equation}\label{eq:functional_H_dot2}
\begin{split}
\dot{\mathcal{H}}(x_t,\dot{x}_t)
=
w_0(x_t) 
 \F(x_t,\dot{x}_t) + 
\sum_{j=1}^l w_j(x_t)\dot{x}_t(-\tau_j) % \\
+ \int_{-\sigma_1}^{-\sigma_2} w_{\rm d}(x_t,\vartheta) \dot{x}_t(\vartheta) \, {\rm d}\vartheta.
\end{split}
\end{equation}

%%%%%%%%%%%%%%%%%%%%%%%%%%%%%%%%%%%%%%%%%%%%%%%%%%%%%%%%%%%%%%%%%%%%%%%%%%%%%%%%%%%%%%%%%%%%
\section{Safety-critical Control of time delay systems}\label{sec:safety_crit_cont}
%%%%%%%%%%%%%%%%%%%%%%%%%%%%%%%%%%%%%%%%%%%%%%%%%%%%%%%%%%%%%%%%%%%%%%%%%%%%%%%%%%%%%%%%%%%%

Building upon the framework establishing safety for dynamical systems with time delays, we now extend this approach to control systems with time delays.
Similarly to how safety functions were extended to control barrier functions in Section~\ref{sec:delay_free}, in this section we extend safety functionals to control barrier functionals, and use them as tool for safety-critical controller synthesis.

Let us consider the following affine control system with state delay:
\begin{equation}\label{eq:affine_sys_delay}
\dot{x}(t)=\F(x_t) + \G(x_t)u(t),
\end{equation}
where $x \in \R^n$ is the state, $x_t \in \B$ is the state history defined in \eqref{eq:state_function}, $\dot{x}_t \in \Q$ is its derivative, $u \in \R^m$ is the input, while $\F:\B \to \R^n$ and $\G:\B \to \R^{n \times m}$  are locally Lipschitz continuous functionals 
in their first argument and continuous in their second.
The state derivative along this system can also be expressed as:
\begin{equation}\label{eq:state_derivative_control}
\dot{x}_t(\vartheta)=
\begin{cases}
\F(x_t) + \G(x_t)u(t), & {\rm if} \; \vartheta = 0,\\
\dot{x}(t+\vartheta) & {\rm if} \; \vartheta \in[-\tau, 0).
\end{cases}
\end{equation}
We seek to design a locally Lipschitz continuous controller ${\K \colon \B \times \Q \to \R^m}$, ${u=\K(x_t,\dot{x}_t)}$,
and enforce the closed-loop system:
\begin{equation}\label{eq:closedloop_delay}
\dot{x}(t) = \F(x_t) + \G(x_t) \K(x_t,\dot{x}_t)
\end{equation}
to be safe w.r.t. the set $\mathcal{S}$ in~\eqref{eq:safesetcond_delay}, that is, $\mathcal{S}$ is forward invariant w.r.t.~\eqref{eq:closedloop_delay}.
The overall right-hand side of the closed-loop system ~\eqref{eq:closedloop_delay} should satisfy the same restrictions as $\F$ in~\eqref{eq:neutral_Lips} to have existence and uniqueness guarantee.

To design a control input that guarantees the system to be safe motivates the introduction of {\em control barrier functionals}.
\begin{definition}[\textbf{Control Barrier Functional, CBFal}]\label{def:CBFal}
A continuously Fr\'{e}chet differentiable functional ${\mathcal{H} \colon \B \to \R}$ is a \textbf{control barrier functional} (CBFal) for \eqref{eq:affine_sys_delay} on $\mathcal{S}$ defined by~\eqref{eq:safesetcond_delay},
if there exists $\alpha \in \Keinf$ such that $\forall x_t \in \B$:
\begin{equation}\label{eq:CBFal_condition}
\sup_{u\in \R^m} \dot{\mathcal{H}}(x_t,\dot{x}_t,u) > - \alpha \big( \mathcal{H}(x_t) \big),
\end{equation}
where
\begin{equation}\label{eq:H_Lie_deriv}
 \dot{ \mathcal{H} }(x_t,\dot{x}_t,u) =
\mathcal{L}_\F \mathcal{H} (x_t,\dot{x}_t)
 + 
\mathcal{L}_\G \mathcal{H} (x_t)\,u
\end{equation}
is the derivative of $\mathcal{H}$ along \eqref{eq:affine_sys_delay}, given by $\dot{x}_t$ in~\eqref{eq:state_derivative_control} and the functionals 
${\mathcal{L}_\F \mathcal{H} \colon \B \to \R}$ 
and ${\mathcal{L}_\G \mathcal{H} \colon \B \to \R}$.
\end{definition}

Note that the derivative $\dot{\mathcal{H}}$ is still obtained by the linear functional in~\eqref{eq:thm_bounded_var}, however, $\dot{x}_t$ is now given by~\eqref{eq:state_derivative_control} and it involves $u(t)$.
Therefore, $\dot{\mathcal{H}}$ in~\eqref{eq:CBFal_condition} depends on $u$ as an affine function.
Analogously to~\eqref{eq:functional_H_dot}, the time derivative of $\mathcal{H}$ is expressed as follows for the case of multiple point delays and an additional distributed delay:
\begin{equation}\label{eq:Lie_derivatives}
\begin{split}
 \mathcal{L}_\F  \mathcal{H} (x_t,\dot{x}_t) = & w_0(x_t) \F(x_t) + 
\sum_{j=1}^l w_j(x_t)\dot{x}_t(-\tau_j) % \\&
+ \int_{-\sigma_1}^{-\sigma_2} w_{\rm d}(x_t,\vartheta) \dot{x}_t(\vartheta) \, {\rm d}\vartheta,
\\
  \mathcal{L}_\G  \mathcal{H} (x_t)= & w_0(x_t)\G(x_t). 
\end{split}
\end{equation}

With the definition of CBFal we state our main result to ensure safety for systems with state delay by extending Theorem~\ref{thm:safe_RFDE}.
\begin{theorem}\label{thm:safety_delay}
\textit{
If $\mathcal{H}$ is a CBFal for \eqref{eq:affine_sys_delay} on $\mathcal{S}$ defined by~\eqref{eq:safesetcond_delay}, then any locally Lipschitz continuous controller $\K \colon \B \times \Q \to \R^m$, $u=\K(x_t,\dot{x}_t)$ satisfying:
\begin{equation} \label{eq:safety_delay}
\dot{\mathcal{H}}\big(x_t,\dot{x}_t,\K(x_t,\dot{x}_t)\big) \geq - \alpha \big( \mathcal{H}(x_t) \big).
\end{equation}
${\forall x_t \in \mathcal{S}}$ renders set $S$ forward invariant w.r.t.~\eqref{eq:closedloop_delay}.
}
\end{theorem}
\begin{proof}[Proof]
Definition~\ref{def:CBFal} ensures that controller $\K$ exists. Then considering the closed-loop dynamics in~\eqref{eq:closedloop_delay}, the set $\mathcal{S}$ is forward invariant according to Theorem~\ref{thm:safe_RFDE}.
\end{proof}

This result motivates the construction of pointwise optimal safety-critical controllers that  use the nearest safe action to a nominal but potentially unsafe controller.
\begin{corollary}\label{cor:min_norm}
\textit{
Given a CBFal $\mathcal{H}$ and a locally Lipschitz continuous desired controller ${\K_{\rm des}\colon \B \times \Q  \to \R^m}$, $u_{\rm des}=\K_{\rm des}(x_t,\dot{x}_t)$ under the condition~\eqref{eq:neutral_Lips}, the following quadratic program (QP) yields a controller ${\K\colon \B \times \Q \to \R^m}$, $u=\K(x_t,\dot{x}_t)$ that renders set $\mathcal{S}$ in \eqref{eq:safesetcond_delay} forward invariant w.r.t.~\eqref{eq:closedloop_delay}:
\begin{equation}\label{eq:QP_delay}
\begin{split}
\K(x_t,\dot{x}_t) =
\underset{u \in \R^m}{\operatorname{argmin}} & \quad \frac{1}{2}\|u - \K_{\rm des}(x_t,\dot{x}_t)\|_2^2  \\
% \mathrm{s.t.} & \quad \eqref{eq:safety_cond_func}.
\mathrm{s.t.} & \quad \dot{\mathcal{H}}(x_t,\dot{x}_t,u) \geq - \alpha \big( \mathcal{H}(x_t) \big).
\end{split}
\end{equation}
% This defines the optimal safe control input ${u= \K(x_t,\dot{x}_t)}$ {\em implicitly} by the safe controller $\K:\B \times \B \to \R^m$ which is the functional of $x_t$ and $\dot{x}_t$.
Furthermore, the explicit solution of~\eqref{eq:QP_delay} can be found by the Karush–Kuhn–Tucker (KKT) \cite{Boyd2004} conditions as:
\begin{equation}\label{eq:min-norm}
\K(x_t,\dot{x}_t)= 
\begin{cases}
\K_{\rm des}(x_t,\dot{x}_t) & {\rm if} \; \phi(x_t,\dot{x}_t) \geq 0,\\
\K_{\rm des}(x_t,\dot{x}_t) -\frac{\phi(x_t,\dot{x}_t) \phi_0^\top(x_t)}{\phi_0(x_t) \phi_0^\top(x_t}  & {\rm otherwise},
\end{cases}
\end{equation}
% where the functionals $\phi \colon \B \times \B \to \R$ and $\phi_0 \colon \B \to \R^m$ read
% \begin{equation}\label{eq:PHI_PHI0}
% \begin{split}
% \phi(x_t,\dot{x}_t) =& \mathcal{L}_\F  \mathcal{H}(x_t,\dot{x}_t)  + \mathcal{L}_\G  \mathcal{H}(x_t) \, \K_{\rm des}(x_t,\dot{x}_t) + \alpha \left( \mathcal{H}(x_t) \right) \\
% \phi_0(x_t) =& \mathcal{L}_\G  \mathcal{H}(x_t).
% \end{split}
% \end{equation}
% $\phi(x_t,\dot{x}_t) = \mathcal{L}_\F  \mathcal{H}(x_t,\dot{x}_t)  + \mathcal{L}_\G  \mathcal{H}(x_t) \, \K_{\rm des}(x_t) + \alpha \left( \mathcal{H}(x_t) \right)$ and $\phi_0(x_t) = \mathcal{L}_\G  \mathcal{H}(x_t)$.
where $\phi(x_t,\dot{x}_t) = \mathcal{L}_\F  \mathcal{H}(x_t,\dot{x}_t)  + \mathcal{L}_\G  \mathcal{H}(x_t) \, \K_{\rm des}(x_t,\dot{x}_t) + \alpha \left( \mathcal{H}(x_t) \right)$
and $\phi_0(x_t) = \mathcal{L}_\G  \mathcal{H}(x_t)$.
}
\end{corollary}
\noindent The derivation of \eqref{eq:min-norm} is detailed in Appendix~\ref{sec:appdx_KKT}.
% The proof of this theorem is discussed in the Appendix~\ref{sec:appdx_KKT}.
This allows operating the nominal controller when it is safe ($\phi \geq 0$), and modifies the input to keep the system safe otherwise ($\phi < 0$).

\begin{remark}[(Control of neutral systems)]\label{rem:neutral_control}
Note that while the control system \eqref{eq:affine_sys_delay} contains retarded terms (functionals of $x_t$) only, the closed-loop system \eqref{eq:closedloop_delay} is neutral (includes $\dot{x}_t$).
Therefore, one may extend the theory to neutral control systems of the form:
\begin{equation}\label{eq:affine_sys_neutral}
\dot{x}(t)=\F(x_t,\dot{x}_t) + \G(x_t,\dot{x}_t)u(t),
\end{equation}
where $\F:\B \times \Q \to \R^n$ and $\G:\B \times \Q  \to \R^{n \times m}$  are locally Lipschitz continuous functionals 
in their first argument and continuous in their second.
The corresponding closed-loop system:
\begin{equation}\label{eq:closedloop_neutral}
\dot{x}(t) = \F(x_t,\dot{x}_t) + \G(x_t,\dot{x}_t) \K(x_t,\dot{x}_t)
\end{equation}
is still neutral, and its right-hand side should satisfy the same restrictions as $\F$ in~\eqref{eq:neutral_Lips} to have existence and uniqueness guarantee.
For control design, we further require that $\F(x_t,\dot{x}_t)$ and $\G(x_t,\dot{x}_t)$ do not include $\dot{x}_t(0)=\dot{x}(t)$, i.e., that \eqref{eq:affine_sys_neutral} expresses $\dot{x}(t)$ explicitly (not implicitly).
Then, safety-critical controllers can still be synthesized using \eqref{eq:safety_delay}, where the derivative of $\mathcal{H}$ is:
\begin{equation}\label{eq:H_Lie_deriv_neutral}
 \dot{ \mathcal{H} }(x_t,\dot{x}_t,u) =
\mathcal{L}_\F \mathcal{H} (x_t,\dot{x}_t)
 + 
\mathcal{L}_\G \mathcal{H} (x_t,\dot{x}_t)\,u
\end{equation}
cf.~\eqref{eq:H_Lie_deriv}, with
${\mathcal{L}_\F \mathcal{H} \colon \B \times \Q \to \R}$ 
and ${\mathcal{L}_\G \mathcal{H} \colon \B \times \Q \to \R}$
given by:
\begin{equation}\label{eq:Lie_derivatives_neutral}
\begin{split}
 \mathcal{L}_\F  \mathcal{H} (x_t,\dot{x}_t) = & w_0(x_t) \F(x_t,\dot{x}_t) + 
\sum_{j=1}^l w_j(x_t)\dot{x}_t(-\tau_j) % \\&
+ \int_{-\sigma_1}^{-\sigma_2} w_{\rm d}(x_t,\vartheta) \dot{x}_t(\vartheta) \, {\rm d}\vartheta,
\\
  \mathcal{L}_\G  \mathcal{H} (x_t,\dot{x}_t)= & w_0(x_t)\G(x_t,\dot{x}_t). 
\end{split}
\end{equation}
cf.~\eqref{eq:Lie_derivatives}.
Note that in this case $\mathcal{L}_\F  \mathcal{H}$ may not be independent of $\dot{x}_t$ even if the conditions in Remark~\ref{rem:simplification} hold due to the occurrence of $\dot{x}_t$ in $\F$.
Furthermore, $\mathcal{L}_\G  \mathcal{H}$ also depends on $\dot{x}_t$ through $\G$.
\end{remark}

\begin{corollary}\label{cor:safety_neutral}
\textit{
If $\mathcal{H}$ is a CBFal for \eqref{eq:affine_sys_neutral} on $\mathcal{S}$ defined by~\eqref{eq:safesetcond_delay}, then any locally Lipschitz continuous controller $\K \colon \B \times \Q \to \R^m$, $u=\K(x_t,\dot{x}_t)$ satisfying \eqref{eq:safety_delay}
${\forall x_t \in \mathcal{S}}$ renders set $S$ forward invariant w.r.t.~\eqref{eq:closedloop_neutral}.
}
\end{corollary}

In the presence of time delay, the relative degree is affected by the delay itself, which is detailed next.
For simplicity, we omit further discussions on neutral control systems, and the rest of the paper addresses the retarded control system \eqref{eq:affine_sys_delay}.

\subsection{Delay-induced Higher Relative Degree} 
\label{sec:rel_deg}

The CBFal condition~\eqref{eq:CBFal_condition} sufficiently holds if ${\mathcal{L}_\G \mathcal{H}(x_t) \neq 0}$ for all ${\forall x_t \in \B}$ and ${\forall \dot{x}_t \in \Q}$.
We refer to this as $\mathcal{H}$ having {\em relative degree}~$1$.
Relative degree is an important concept for time delay systems, too, which motivates the extension of Definition~\ref{def:reldeg}.
\begin{definition}[\textbf{Relative Degree of Functional}]\label{def:reldeg_delay}
Functional ${\mathcal{H} \colon \B \to \R}$ has \textbf{relative degree} $r$ (where ${r \in \mathbb{Z}}$, ${r \geq 1}$) w.r.t.~\eqref{eq:affine_sys_delay} if it is $r$ times continuously Fr\'{e}chet differentiable and satisfies the two conditions below.
\begin{enumerate}[I.]
\item For ${r \geq 2}$, the Fr\'{e}chet derivative of ${\mathcal{L}_\F^k \mathcal{H}(x_t,\dot{x}_t)}$ w.r.t. $\dot{x}_t$ is zero ${\forall x_t \in \B}$ and ${\forall k \in \{1, \ldots, r-1 \}}$.
That is, ${\mathcal{L}_\F^k \mathcal{H}(x_t,\dot{x}_t)}$ does not depend on $\dot{x}_t$, only on $x_t$, which we denote shortly as ${\mathcal{L}_\F^k \mathcal{H}(x_t)}$.
\item The following holds ${\forall x_t \in \B}$:
\begin{align}
\begin{split}
    \mathcal{L}_\G \mathcal{L}_\F^{r-1} \mathcal{H}(x_t) & \neq 0, \\
    \mathcal{L}_\G \mathcal{L}_\F^k \mathcal{H}(x_t) & = 0, \quad {\rm for} \; r \geq 2, \; k \in \{0, \ldots, r-2 \}, \\
\end{split}
\end{align}
where ${\mathcal{L}_\G \mathcal{L}_\F^0 \mathcal{H}(x_t) = \mathcal{L}_\G \mathcal{H}(x_t)}$ and the second condition only applies for ${r \geq 2}$.
\end{enumerate}
\end{definition}

Condition I is specific to time delay systems and does not have a delay-free counterpart in Definition~\ref{def:reldeg}.
This condition is imposed because otherwise for higher relative degree, $r \geq 2$, higher time derivatives of $\mathcal{H}$ could include higher derivatives of $x_t$.
Synthesizing a controller in this case could lead to a closed-loop system where the rate of change of state depends on past values of higher derivatives of the state, which is called by advanced functional differential equation (AFDE). We demonstrate this by an example in Section~\ref{sec:examp}. Advanced type equations are rarely used in engineering applications due to their inverted causality problem \cite{InspergerStepan2011}.

Definition~\ref{def:reldeg_delay} together with~\eqref{eq:Lie_derivatives} leads to the following observations.
\begin{itemize}
    \item For {\em relative degree~$1$}, ${\mathcal{L}_\G \mathcal{H}(x_t) = w_0(x_t)\G(x_t) \neq 0}$ that implies ${w_0(x_t) \neq 0}$.
This means that $\mathcal{H}(x_t)$ contains $x(t)$
as discussed in Remark~\ref{rem:atomic}, i.e., the {\em CBFal includes the present state $x(t)$}.
    \item For {\em relative degree~$2$}, ${\mathcal{L}_\F \mathcal{H}(x_t)}$ is independent of $\dot{x}_t$.
Per Remark~\ref{rem:simplification}, this implies ${w_j(x_t)=0}$, ${\forall x_t \in \B}$, ${\forall j \geq 1}$, i.e., the {\em CBFal excludes states $x(t-\tau_j)$ with point delay}.
Moreover, ${\mathcal{L}_\G \mathcal{H}(x_t) = w_0(x_t)\G(x_t) = 0}$ and ${\mathcal{L}_\G \mathcal{L}_\F \mathcal{H}(x_t) \neq 0}$.
The former occurs if ${w_0(x_t) = 0}$, i.e., the CBFal does not contain the present state $x(t)$ only a distributed delay term.
Alternatively, the expression of $\G$ may also cause ${\mathcal{L}_\G \mathcal{H}(x_t) = 0}$, even if ${w_0(x_t) \neq 0}$ and the CBFal includes the present state $x(t)$.
    \item There is {\em no valid relative degree} if ${w_0(x_t) = 0}$ and $w_j(x_t) \neq 0$ for some $j \geq 1$, i.e., when the {\em CBFal excludes the present state $x(t)$ but includes the past state ${x(t-\tau_j)}$}. That is the past of a system cannot be rendered safe.
\end{itemize}
\noindent As such, distributed delays in $\mathcal{H}$ may induce higher relative, whereas point delays may lead to no valid relative degree. 

For higher relative degree problems in delay-free systems, controllers can be synthesized by constructing a relative degree~$1$ CBF from $h$ and its derivatives; see the methods in \cite{Nguyen2016, Xiao2019, sarkar2020highrelative, Wang2020LearningCB}.
This process can be extended to time delay systems, which we demonstrate for relative degree~$2$ (the process of extended to higher relative degrees is similar, but more notationally intensive, so this case suffices to illustrate the process without loss of generality).
We consider ${\mathcal{L}_\G  \mathcal{H}(x_t) = 0}$ and that  ${\mathcal{L}_\F  \mathcal{H}}$ does not contain terms of $\dot{x}_t$.
We introduce the {\em extended control barrier functional}:
\begin{equation}\label{eq:ext_safety_func}
\begin{split}
    &  \mathcal{H}_{\rm e}(x_t)  = %\\ & 
    \underbrace{\dot{ \mathcal{H}}(x_t)}_
    { \mathcal{L}_\F \mathcal{H}(x_t)}
    % = \underbrace{\langle \frac{\partial \mathcal{H}}{\partial x_t} , \F(x_t) \rangle}
+\alpha \big( \mathcal{H}(x_t)\big),
\end{split}
\end{equation}
that has relative degree~$1$ if $\mathcal{H}$ has relative degree~$2$ since ${\mathcal{L}_\G \mathcal{L}_\F \mathcal{H}(x_t) \neq 0}$.
Its 0-superlevel set is denoted by:
\begin{equation}\label{eq:safesetcond_extended}
    \mathcal{S}_{\rm e} = \big\{x_t \in \B : \mathcal{H}_{\rm e}(x_t) \geq 0 \big\}.
\end{equation}

\begin{definition}[\textbf{Extended Control Barrier Functional}]\label{def:eCBFal}
Let ${\mathcal{H} \colon \B \to \R}$ be a twice continuously Fr\'{e}chet differentiable functional with continuously differentiable ${\alpha \in \Keinf}$, ${\mathcal{L}_\G  \mathcal{H}(x_t) = 0}$ and continuously differentiable ${\mathcal{L}_\F \mathcal{H}}$ excluding terms of $\dot{x}_t$.
Then functional ${\mathcal{H}_{\rm e} \colon \B \to \R}$ defined by~\eqref{eq:ext_safety_func} is a \textbf{extended control barrier functional} (extended CBFal) for \eqref{eq:affine_sys_delay} on $\mathcal{S} \cap \mathcal{S}_{\rm e}$ defined by~\eqref{eq:safesetcond_delay} and~\eqref{eq:safesetcond_extended},
if there exists ${\alpha_{\rm e} \in \Keinf}$ such that ${\forall x_t \in \B}$:
\begin{equation}\label{eq:CBFal_condition_extended}
\sup_{u\in \R^m} \dot{\mathcal{H}}_{\rm e}(x_t,\dot{x}_t,u) > - \alpha_{\rm e} \big( \mathcal{H}_{\rm e}(x_t) \big),
\end{equation}
where
\begin{equation}\label{eq:He_Lie_deriv}
 \dot{ \mathcal{H} }_{\rm e}(x_t,\dot{x}_t,u) =
 \mathcal{L}^2_\F \mathcal{H}(x_t,\dot{x}_t) 
+ \mathcal{L}_\G  \mathcal{L}_\F  \mathcal{H}(x_t) \,u
+ \alpha' \big( \mathcal{H}(x_t)\big) \mathcal{L}_\F \mathcal{H}(x_t),
\end{equation}
is the derivative of $\mathcal{H}_{\rm e}$ along \eqref{eq:affine_sys_delay}, given by $\dot{x}_t$ in~\eqref{eq:state_derivative_control} and the functionals ${\mathcal{L}^2_\F \mathcal{H} \colon \B \times \Q \to \R}$ and ${\mathcal{L}_\G \mathcal{L}_\F \mathcal{H} \colon \B \to \R}$.
\end{definition}

Note that condition~\eqref{eq:CBFal_condition_extended} sufficiently holds if ${\mathcal{L}_\G  \mathcal{L}_\F  \mathcal{H}(x_t) \neq 0}$, ${\forall x_t \in \B}$, i.e., in case $\mathcal{H}$ has relative degree~$2$.
With this definition we can state the theorem to ensure safety for systems with (potentially delay-induced) relative degree~$2$.
\begin{theorem}\label{thm:safety_delay_extended}
\textit{
If $\mathcal{H}_{\rm e}$ is an extended CBFal for \eqref{eq:affine_sys_delay} on $\mathcal{S} \cap \mathcal{S}_{\rm e}$ defined by~\eqref{eq:safesetcond_delay} and~\eqref{eq:safesetcond_extended}, then any locally Lipschitz continuous controller ${\K \colon \B \times \Q \to \R^m}$, ${u=\K(x_t,\dot{x}_t)}$ satisfying:
\begin{equation}
\dot{ \mathcal{H}}_{\rm e}(x_t,\dot{x}_t,\K(x_t,\dot{x}_t)) \geq -\alpha_{\rm e} \big( \mathcal{H}_{\rm e}(x_t)\big),
\end{equation}
${\forall x_t \in \mathcal{S} \cap \mathcal{S}_{\rm e}}$ renders set ${\mathcal{S} \cap \mathcal{S}_{\rm e}}$ forward invariant w.r.t.~\eqref{eq:closedloop_delay}.
}
\end{theorem}
\begin{proof}[Proof]
We prove that ${x_0 \in \mathcal{S} \cap \mathcal{S}_{\rm e} \implies x_t \in \mathcal{S} \cap \mathcal{S}_{\rm e}}$, ${\forall t \geq 0}$.
Note that ${x_0 \in \mathcal{S} \cap \mathcal{S}_{\rm e}}$ yields both ${x_0 \in \mathcal{S}}$ and ${x_0 \in \mathcal{S}_{\rm e}}$.
By Theorem~\ref{thm:safety_delay}, we have ${x_0 \in \mathcal{S}_{\rm e} \implies x_t \in \mathcal{S}_{\rm e}}$, ${\forall t \geq 0}$.
This means ${\mathcal{H}_{\rm e}(x_t) \geq 0}$ holds, i.e., ${\dot{ \mathcal{H}}(x_t) \geq - \alpha \big( \mathcal{H}(x_t) \big)}$ based on \eqref{eq:ext_safety_func}.
Therefore, applying Theorem~\ref{thm:safe_RFDE} (or more precisely, Corollary~\ref{cor:safety_neutral}) yields ${x_0 \in \mathcal{S} \cap \mathcal{S}_{\rm e} \implies x_t \in \mathcal{S} \cap \mathcal{S}_{\rm e}}$, ${\forall t \geq 0}$.
\end{proof}

With this theorem, one can design a pointwise optimal controller similarly as in \eqref{eq:min-norm}.

\begin{corollary}\label{cor:min_norm_extended}
\textit{
Given an extended CBFal $\mathcal{H}_{\rm e}$ and a locally Lipschitz continuous desired controller ${\K_{\rm des}\colon \B \times \Q \to \R^m}$, ${u_{\rm des}=\K_{\rm des}(x_t,\dot{x}_t)}$, the following quadratic program (QP) yields a controller ${\K\colon \B \times \Q \to \R^m}$, ${u=\K(x_t,\dot{x}_t)}$ that renders set ${\mathcal{S} \cap \mathcal{S}_{\rm e}}$ in \eqref{eq:safesetcond_delay} and \eqref{eq:safesetcond_extended} forward invariant w.r.t.~\eqref{eq:closedloop_delay}:
\begin{equation}\label{eq:QP_delay_extended}
\begin{split}
\K(x_t,\dot{x}_t) =
\underset{u \in \R^m}{\operatorname{argmin}} & \quad \frac{1}{2}\|u - \K_{\rm des}(x_t,\dot{x}_t)\|_2^2  \\
\mathrm{s.t.} & \quad \dot{\mathcal{H}}_{\rm e}(x_t,\dot{x}_t,u) \geq - \alpha_{\rm e} \big( \mathcal{H}_{\rm e}(x_t) \big).
\end{split}
\end{equation}
Furthermore, the explicit solution of~\eqref{eq:QP_delay_extended} can be found by the Karush–Kuhn–Tucker (KKT) \cite{Boyd2004} conditions as:
\begin{equation}\label{eq:min-norm_extended}
\K(x_t,\dot{x}_t)= 
\begin{cases}
\K_{\rm des}(x_t,\dot{x}_t) & {\rm if} \; \phi_{\rm e}(x_t,\dot{x}_t) \geq 0,\\
\K_{\rm des}(x_t,\dot{x}_t) -\frac{\phi_{\rm e}(x_t,\dot{x}_t) \phi_{0{\rm e}}^\top(x_t)}{\phi_{0{\rm e}}(x_t) \phi_{0{\rm e}}^\top(x_t)}  & {\rm otherwise},
\end{cases}
\end{equation}
where ${\phi_{\rm e}(x_t,\dot{x}_t) =\mathcal{L}^2 _\F  \mathcal{H}(x_t,\dot{x}_t)  + \mathcal{L}_\G  \mathcal{L}_\F \mathcal{H}(x_t) \, \K_{\rm des}(x_t,\dot{x}_t) + \alpha' \big( \mathcal{H}(x_t)\big) \mathcal{L}_\F \mathcal{H}(x_t) + \alpha_{\rm e}\left( \mathcal{H}_{\rm e}(x_t) \right)}$
and ${\phi_{0{\rm e}}(x_t) =  \mathcal{L}_\G \mathcal{L}_\F \mathcal{H}(x_t)}$.
}
\end{corollary}
The detailed derivation of \eqref{eq:min-norm_extended} is similar to~\eqref{eq:min-norm} which is discussed in Appendix~\ref{sec:appdx_KKT}.

%%%%%%%%%%%%%%%%%%%%%%%%%%%%%%%%%%%%%%%%%%%%%%%%%%%%%%%%%%%%%%%%%%%%%%%%%%%%%%%%%%%%%%%%%%%%
\section{Examples and application}\label{sec:examp}
%%%%%%%%%%%%%%%%%%%%%%%%%%%%%%%%%%%%%%%%%%%%%%%%%%%%%%%%%%%%%%%%%%%%%%%%%%%%%%%%%%%%%%%%%%%%

Now we apply the theoretical constructions of this paper on demonstrative examples.
First, we address systems with point delay, second, we discuss a scalar control system with different types of delay.

\begin{eexample} 
Consider the following affine control system with point delay $\tau>0$:
\begin{equation}\label{eq:example_DDE}
\dot{x}(t) = f \big( x(t), x(t-\tau) \big) + g \big( x(t), x(t-\tau) \big) u(t),
%=-x(t-\tau) + x^3(t),
\end{equation}
where $x \in \R^n$, $u \in \R^m$, while $f:\R^n \times \R^n \to \R^n$ and $g:\R^n \times \R^n \to \R^{n \times m}$ are locally Lipschitz continuous.
To keep this system safe, we seek to enforce:
\begin{equation}
h \big( x(t), x(t-\tau) \big) \geq 0, \quad \forall t \geq 0,
\end{equation}
along the solutions of the corresponding closed control loop, where $h:\R^n \times \R^n \to \R$ is continuously differentiable.

System~\eqref{eq:example_DDE} and the corresponding safe set can be rewritten in the form as~\eqref{eq:affine_sys_delay} and~\eqref{eq:safesetcond_delay} with the functionals $\F \colon \B \to \R^n$, $\G \colon \B \to \R^{n \times m}$ and $\mathcal{H} \colon \B \to \R$ defined by:
\begin{equation}
\begin{split}
    \F(x_t) & = f \big( x_t(0), x_t(-\tau) \big),\\
    \G(x_t) & = g \big( x_t(0), x_t(-\tau) \big),\\
    \mathcal{H}(x_t) & = h \big( x_t(0), x_t(-\tau) \big).
\end{split}
\end{equation}
This choice of functional $\mathcal{H}$ is a special case of that in \eqref{eq:safe_funnal_gen_exa} in Example~\ref{example:functional}.
The time derivative of $\mathcal{H}$ becomes:
\begin{equation}
\begin{split}
\dot{\mathcal{H}}(x_t,\dot{x}_t,u) = 
\nabla_0 h \big( x_t(0), x_t(-\tau) \big)  \Big( f \big( x_t(0), x_t(-\tau) \big) + g \big( x_t(0), x_t(-\tau) \big) u(t) \Big) 
+ \nabla_1 h \big( x_t(0), x_t(-\tau) \big)  \dot{x}_t(-\tau),
\end{split}
\end{equation}
cf.~\eqref{eq:safe_funnal_der_exa}, where $\nabla_j$ denotes the gradient with respect to the $j$-th vector-valued variable.
This expression corresponds to~\eqref{eq:H_Lie_deriv} with:
\begin{equation}
\begin{split}
\mathcal{L}_\F \mathcal{H}(x_t,\dot{x}_t) =& 
\nabla_0 h \big( x_t(0), x_t(-\tau) \big)  f \big( x_t(0), x_t(-\tau) \big) + 
\nabla_1 h \big( x_t(0), x_t(-\tau) \big)  \dot{x}_t(-\tau),\\
\mathcal{L}_\G \mathcal{H}(x_t) = & 
\nabla_0 h \big( x_t(0), x_t(-\tau) \big)  g \big( x_t(0), x_t(-\tau) \big),
\end{split}
\end{equation}
and the weights
${w_0(x_t)=\nabla_0 h \big( x_t(0), x_t(-\tau) \big)}$,
${w_1(x_t)=\nabla_1 h \big( x_t(0), x_t(-\tau) \big)}$ and
${w_{\rm d}(x_t,\vartheta)=0}$
in \eqref{eq:Lie_derivatives}.
Thereafter, safety-critical controllers can be designed based on \eqref{eq:safety_delay} in Theorem~\ref{thm:safety_delay}, such the quadratic program-based controller \eqref{eq:min-norm}.
\end{eexample} 

\begin{eexample}
Consider the scalar control system:
\begin{equation}\label{eq:scalarsystem}
\dot{x}(t)= 
\underbrace{
x^3(t)
}_{\F(x_t)}
+ 
\underbrace{
x(t-\tau)
}_{\G(x_t)}
u(t),
\end{equation}
where ${x \in \R}$, ${u \in \R}$.
Notice that this system is not forward complete and for some controllers (including $u(t) \equiv 0$) the solution may have finite escape time.
Hence, we seek to keep the state within safe bounds.
We compare different types of CBFals, that include delay-free, point delay and distributed delay terms, as special cases of the functional in \eqref{eq:safe_funnal_gen_exa} in Example~\ref{example:functional}.  

{\em Case 1:}
First, we intend to keep the solution $x(t)$ within ${[-1,\, 1]}$ by constructing the CBFal with delay-free term as:
\begin{equation}\label{eq:h_scalarsystem_1}
\mathcal{H}(x_t) = 1 - x^2(t).
\end{equation}
This corresponds to \eqref{eq:safe_funnal_gen_exa} with ${h(s,\dots)=1-s^2}$, and the time derivative of the CBFal reads:
\begin{equation}
\dot{\mathcal{H}}(x_t,\dot{x}_t,u)=-2x(t)\left(x^3(t)+x(t-\tau)u\right)
\end{equation}
with weights ${w_0(x_t)=-2x(t)}$, ${w_j(x_t)=0}$, ${w_{\rm d}(x_t,\vartheta)=0}$, while the expressions
in \eqref{eq:Lie_derivatives} become:
\begin{equation}
\begin{split}
    \mathcal{L}_\F \mathcal{H}(x_t,\dot{x}_t) = & -2 x^4(t), \\
    \mathcal{L}_\G \mathcal{H}(x_t) = & -2x(t) x(t-\tau).
\end{split}
\end{equation}

We apply Theorem~\ref{thm:safety_delay} to ensure safety and implement the QP-based controller~\eqref{eq:min-norm} using the desired controller ${\K_{\rm des}(x_t,\dot{x}_t)=0}$ and a linear class-$\Keinf$ function ${\alpha(r) = \gamma r}$ with ${\gamma>0}$,
that results in:
\begin{equation}\label{}
\K(x_t,\dot{x}_t)= 
\begin{cases}
0 & {\rm if} \; -2x^4(t)+\gamma \left(1-x^2(t) \right) \geq 0,\\
\frac{-2x^4(t)+\gamma \left(1-x^2(t)\right)}{2x(t) x(t-\tau)}  & {\rm otherwise}.
\end{cases}
\end{equation}
Substitution back into \eqref{eq:scalarsystem} leads to the closed-loop system that forms the differential equation:
\begin{equation}\label{eq:closed_loop_scalar1}
\dot{x}(t)= 
\begin{cases}
x^3(t) & {\rm if} \; -2x^4(t)+\gamma \left(1-x^2(t) \right) \geq 0,\\
\frac{\gamma \left(1-x^2(t)\right)}{2x(t)}  & {\rm otherwise}.
\end{cases}
\end{equation}
It is important that although in this example the delayed term dropped from the closed-loop system,
this is not the case in general, and it depends on the form of \eqref{eq:scalarsystem} and $\mathcal{H}$.

The performance of this controller is demonstrated by  numerically integrating~\eqref{eq:closed_loop_scalar1}.
Simulation results are plotted in the first column of Fig.~\ref{fig:scalar} (panels (a),(d),(g),(j))  for $\tau=1$, ${\gamma=1}$ and initial condition
${x_0(\vartheta)=0.4}$, ${\vartheta \in [-\tau,0]}$. The panel (a) shows the evolution of the state, and the panel (d) indicates the corresponding control input.
As the state gets close to the safe set boundary (indicated by green line), the controller needs to intervene by deviating from the desired zero input and forces the system to evolve within the safe set.
Intervention starts when the state reaches the switching surface ${-2x^4(t)+\gamma \left(1-x^2(t) \right) = 0}$ (see the dashed line in panel (a)).
Panels (g), (j) at the bottom indicate that safety is successfully maintained as $\mathcal{H}$ is positive for all time while the trajectory in the corresponding phase portrait is kept within the safe set.

{\em Case 2:}
Next, we consider another safe set and intend to keep the squared mean of the solution $x(t)$ and its delayed value ${x(t-\tau)}$ below $1$
by constructing the CBFal:
\begin{equation}\label{eq:h_scalarsystem}
\mathcal{H}(x_t) = 1 - \frac{1}{2}\left(x^2(t) + x^2(t-\tau)\right),
\end{equation}
corresponding to the nonlinear function $h$ in \eqref{eq:safe_funnal_gen_exa} with ${h(s_1,s_2,\dots)=1-\frac{1}{2}(s_1^2+s_2^2)}$.
The derivative of the CBFal becomes:
\begin{equation}
\dot{\mathcal{H}}(x_t,\dot{x}_t,u)=-x(t)\left(x^3(t)+x(t-\tau)u\right) -x(t-\tau)\dot{x}(t-\tau),
\end{equation}
associated with weights ${w_0(x_t)=-x(t)}$, ${w_j(x_t)=-x(t-\tau)}$, ${w_{\rm d}(x_t,\vartheta)=0}$, leading to:
\begin{equation}
\begin{split}
    \mathcal{L}_\F \mathcal{H}(x_t,\dot{x}_t) = & - x^4(t) - x(t-\tau) \dot{x}(t-\tau), \\
    \mathcal{L}_\G \mathcal{H}(x_t) = & -x(t) x(t-\tau),
\end{split}
\end{equation}
cf.~\eqref{eq:Lie_derivatives}.

Implementing the QP-based controller~\eqref{eq:min-norm} with ${\K_{\rm des}(x_t,\dot{x}_t)=0}$ and ${\alpha(r) = \gamma r}$ results in:
\begin{equation}\label{}
\K(x_t,\dot{x}_t)= 
\begin{cases}
0 & {\rm if} \; - x^4(t) - x(t-\tau) \dot{x}(t-\tau)+\gamma \left(1 - \frac{1}{2}\left(x^2(t) + x^2(t-\tau)\right)\right) \geq 0,\\
\frac{- x^4(t) - x(t-\tau) \dot{x}(t-\tau)+\gamma \left(1 - \frac{1}{2}\left(x^2(t) + x^2(t-\tau)\right)\right)}{x(t) x(t-\tau)}  & {\rm otherwise},
\end{cases}
\end{equation}
while the closed-loop system becomes a neutral delay differential equation:
\begin{equation}\label{eq:closed_loop_scalar2}
\dot{x}(t)= 
\begin{cases}
x^3(t) & {\rm if} \; - x^4(t) - x(t-\tau) \dot{x}(t-\tau)+\gamma \left(1 - \frac{1}{2}\left(x^2(t) + x^2(t-\tau)\right)\right) \geq 0,\\
\frac{ - x(t-\tau) \dot{x}(t-\tau)+\gamma \left(1 - \frac{1}{2}\left(x^2(t) + x^2(t-\tau)\right)\right)}{x(t)}   & {\rm otherwise}.
\end{cases}
\end{equation}
The time derivative of the delayed term, ${\dot{x}(t-\tau)}$, appears in the expressions of both the right-hand side and the switching surface.

The middle column of Fig.~\ref{fig:scalar} (panels (b),(e),(h),(k)) plots simulation results for system \eqref{eq:closed_loop_scalar2} with $\tau=1$, $\gamma=1$ and initial conditions
${x_0(\vartheta)=0.4}$, ${\vartheta \in [-\tau,0]}$ and
${\dot{x}_0(\vartheta)=0}$, ${\vartheta \in [-\tau,0)}$. Again, the safety-critical controller keeps the system within the safe set for all time as guaranteed by Theorem~\ref{thm:safety_delay}. This can be clearly seen from the phase portrait in panel (k). 

%%%%%%%%%%%%%%%%%%%%%%%%%%%%%%%%%%%%%%%%
\begin{figure}
\centerline{\includegraphics[width=0.9\columnwidth]{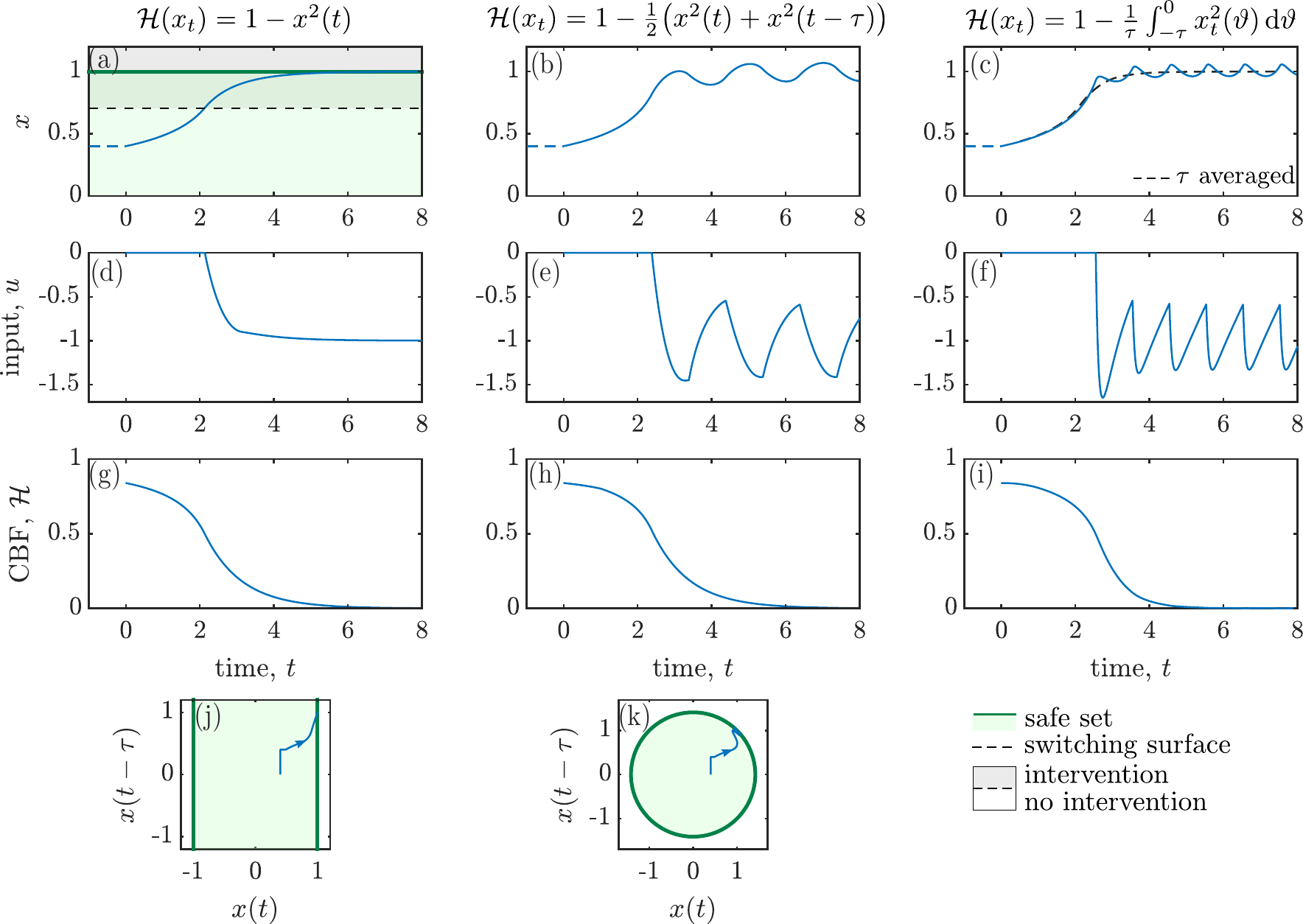}}
\caption{
Safety-critical control of system~\eqref{eq:scalarsystem} with different types of CBFals.
The left column corresponds to {\em Case 1} with delay-free CBFal~\eqref{eq:h_scalarsystem_1}, the middle column shows the results of {\em Case 2} with point delay in the CBFal~\eqref{eq:h_scalarsystem}, while the right column represents {\em Case 3} with distributed delay in the functional~\eqref{eq:distributed_safety_func}.
Panels (a), (b), (c) show the evolution of the state for the three different scenarios, panels (d), (e), (f) plot the value of the synthesized control inputs that maintain safety, panels (g), (h), (i) depict the value of the corresponding functional $\mathcal{H}$, while panels (j), (k) show state space plots.
}
\label{fig:scalar}
\end{figure}
%%%%%%%%%%%%%%%%%%%%%%%%%%%%%%%%%%%%%%%%

{\em Case 3:}
Then, we intend to keep a moving average of the solution (i.e., a root-mean-square average over the delay interval) below $1$
by constructing the CBFal candidate with distributed delay as:
\begin{equation}\label{eq:distributed_safety_func}
\mathcal{H}(x_t) = 1 -  \frac{1}{\tau}\int_{-\tau}^0 x^2(t+\vartheta) \, {\rm d} \vartheta.
\end{equation}
This is a special case of \eqref{eq:safe_funnal_gen_exa} with ${h(\dots,s)=s}$, ${\rho(\vartheta) \equiv \frac{1}{\tau}}$ and ${\kappa(s)=1-s^2}$.
Taking the time derivative of $\mathcal{H}$ and simplifying the integral as in Remark~\ref{rem:simplification}, we get:
\begin{equation}
\begin{split}
\dot{\mathcal{H}}(x_t) = \frac{1}{\tau}\big( x^2(t-\tau) - x^2(t)\big),
\end{split}
\end{equation}
associated with:
\begin{equation}
\begin{split}
    \mathcal{L}_\F \mathcal{H}(x_t) = & \frac{1}{\tau}\big( x^2(t-\tau) - x^2(t)\big), \\
    \mathcal{L}_\G \mathcal{H}(x_t) = & 0.
\end{split}
\end{equation}

Observe that $\dot{\mathcal{H}}$ depends only on $x_t$, as emphasized by the notation ${\dot{\mathcal{H}}(x_t)}$.
Therefore, functional $\mathcal{H}$ in \eqref{eq:distributed_safety_func} is not a valid CBFal because ${\mathcal{L}_\G \mathcal{H}(x_t) = 0}$ and
$u$ does not appear in $\dot{\mathcal{H}}$.
However, since $\dot{x}_t$ does not appear in $\dot{\mathcal{H}}$
either, one may construct the extended CBFal in Definition~\ref{def:eCBFal}.
Based on \eqref{eq:ext_safety_func}, we propose the following extended CBFal:
\begin{equation}
    \mathcal{H}_{\rm e}(x_t) = \frac{1}{\tau}\big( x^2(t-\tau) - x^2(t)\big) + \gamma \left( 1 -  \frac{1}{\tau}\int_{-\tau}^0 x^2(t+\vartheta) \, {\rm d} \vartheta \right),
\end{equation}
whose time derivative is:
\begin{equation}
\dot{\mathcal{H}}_{\rm e}(x_t,\dot{x}_t,u) = \frac{2}{\tau}x(t-\tau) \dot{x}(t-\tau) 
-\frac{2}{\tau}x^4(t) - \frac{2}{\tau}x(t)x(t-\tau) u + \frac{\gamma}{\tau}\big( x^2(t-\tau) - x^2(t)\big),
\end{equation}
cf.~\eqref{eq:He_Lie_deriv}.
This depends on the control input $u$, and it is associated with:
\begin{equation}
\begin{split}
\mathcal{L}^2 _\F \mathcal{H}(x_t,\dot{x}_t) & = -\frac{2}{\tau}x^4(t) + \frac{2}{\tau}x(t-\tau) \dot{x}(t-\tau), \\
\mathcal{L}_\G  \mathcal{L}_\F \mathcal{H}(x_t) & = -\frac{2}{\tau}x(t)x(t-\tau). \\
\end{split}
\end{equation}

Now we use Theorem~\ref{thm:safety_delay_extended} to guarantee safety, and consider a min-norm controller by setting ${\K_{\rm des}(x_t,\dot{x}_t)=0}$, choosing ${\alpha_{\rm e}(r) = \gamma_{\rm e} r}$ with $\gamma_{\rm e} > 0$ and
using \eqref{eq:min-norm_extended}.
This yields the controller:
\begin{equation}
\K(x_t,\dot{x}_t)\!=\! 
\begin{cases}
0 & \!\!{\rm if} \; \frac{-2x^4(t)+2x(t-\tau)\dot{x}(t-\tau)+(\gamma+\gamma_{\rm e}) \big(x^2(t-\tau)-x^2(t)\big) + \gamma_{\rm e}\gamma\left(\tau-\int_{-\tau}^0 x_t^2(\vartheta) \, {\rm d} \vartheta \right)}{\tau}
 \geq 0,\\
\frac{-2x^4(t)+2x(t-\tau)\dot{x}(t-\tau)+(\gamma+\gamma_{\rm e}) \big(x^2(t-\tau)-x^2(t)\big)+\gamma_{\rm e}\gamma\left(\tau-\int_{-\tau}^0 x_t^2(\vartheta) \, {\rm d} \vartheta \right)}{2x(t)x(t-\tau)}
& \!\!{\rm otherwise}, 
\end{cases}
\end{equation}
% Substituting the input back to~\ref{eq:scalarsystem}
and the closed-loop system:
\begin{equation}\label{eq:closed_loop_scalar3}
\dot{x}(t)= 
\begin{cases}
x^3(t) & {\rm if} \; \frac{-2x^4(t)+2x(t-\tau)\dot{x}(t-\tau)+(\gamma+\gamma_{\rm e}) \big(x^2(t-\tau)-x^2(t)\big) + \gamma_{\rm e}\gamma\left(\tau-\int_{-\tau}^0 x_t^2(\vartheta) \, {\rm d} \vartheta \right)}{\tau}
 \geq 0,\\
\frac{2x(t-\tau)\dot{x}(t-\tau)+(\gamma+\gamma_{\rm e}) \big(x^2(t-\tau)-x^2(t)\big)+\gamma_{\rm e}\gamma\left(\tau-\int_{-\tau}^0 x_t^2(\vartheta) \, {\rm d} \vartheta \right)}{2x(t)}
   & {\rm otherwise}, 
\end{cases}
\end{equation}
which is an integro-differential equation with neutral-type delay term.

Simulation results are presented in the right column of Fig.~\ref{fig:scalar} (panels (c),(f),(i)) for system \eqref{eq:closed_loop_scalar3} with ${\tau=1}$, ${\gamma_{\rm e}=1}$, ${\gamma=3}$ and initial conditions
${x_0(\vartheta)=0.4}$, ${\vartheta \in [-\tau,0]}$ and
${\dot{x}_0(\vartheta)=0}$, ${\vartheta \in [-\tau,0)}$.
Although the state variable is sometimes greater than one, the controller
forces the moving average of the solution (dashed curve) to evolve within the safe set for all time, as desired.

To summarize, if the functional $\mathcal{H}$ depends only on the delay-free state, then the QP-based control law depends only on the state $x_t$ and the closed-loop system is an RFDE.
If there is delay in the functional $\mathcal{H}$, then the control law may depend not only on the state $x_t$, but also on its time derivative $\dot{x}_t$, leading to a nonsmooth NFDE as the closed-loop system.
\end{eexample} 

{\em Case 4:}
%[Advanced closed-loop system]
Finally, consider the following functional constructed by the combination of point delay and distributed delay terms:
\begin{equation}\label{}
\tilde{\mathcal{H}}(x_t) = 1 + \frac{1}{2}x^2(t-\tau)- \frac{1}{\tau}\int_{-\tau}^0 x^2(t+\vartheta) \, {\rm d} \vartheta.
\end{equation}
Here we use tilde to emphasise that this functional does not satisfy condition I.~in Definition~\ref{def:reldeg_delay} (and hence does not have a valid relative degree).
Taking the time derivative leads to:
\begin{equation}
\dot{\tilde{\mathcal{H}}}(x_t,\dot{x}_t)
= x(t-\tau)\dot{x}(t-\tau) + \frac{1}{\tau}\big( x^2(t-\tau) - x^2(t)\big),
\end{equation}
associated with:
\begin{equation}
\begin{split}
    \mathcal{L}_\F \tilde{\mathcal{H}}(x_t,\dot{x}_t) = &  x(t-\tau)\dot{x}(t-\tau) + \frac{1}{\tau}\big( x^2(t-\tau) - x^2(t)\big), \\
    \mathcal{L}_\G \tilde{\mathcal{H}}(x_t) = & 0,
\end{split}
\end{equation}
cf.~\eqref{eq:Lie_derivatives}.
One can observe that $\mathcal{L}_\G \tilde{\mathcal{H}}(x_t)=0$ similar to Case 3, leading to the extended functional defined by~\eqref{eq:ext_safety_func} as:
\begin{equation}
    \tilde{\mathcal{H}}_{\rm e}(x_t,\dot{x}_t) = x(t-\tau)\dot{x}(t-\tau) + \frac{1}{\tau}\big( x^2(t-\tau) - x^2(t)\big) + \gamma \left( 1 + \frac{1}{2}x^2(t-\tau) - \frac{1}{\tau}\int_{-\tau}^0 x^2(t+\vartheta) \, {\rm d} \vartheta \right).
\end{equation}
Again, taking its time derivative gives:
\begin{equation}
\dot{\tilde{\mathcal{H}}}_{\rm e}(x_t,\dot{x}_t,\ddot{x}_t,u) = 
 x(t-\tau)\ddot{x}(t-\tau) + \dot{x}^2(t-\tau) + 
\frac{2}{\tau}x(t-\tau) \dot{x}(t-\tau) 
-\frac{2}{\tau}x^4(t) - \frac{2}{\tau}x(t)x(t-\tau) u + \gamma x(t-\tau)\dot{x}(t-\tau) + \frac{\gamma}{\tau}\big( x^2(t-\tau) - x^2(t)\big),
\end{equation}
in which the advanced-type term $\ddot{x}(t-\tau)$ appears, as emphasized by the notation $\dot{\tilde{\mathcal{H}}}_{\rm e}(x_t,\dot{x}_t,\ddot{x}_t,u)$.
Synthesizing a QP-based safe controller from this based on~\eqref{eq:QP_delay_extended} (e.g. with desired controller ${\K_{\rm des}(x_t,\dot{x}_t)=0}$ and linear class-$\Keinf$ function) would result in a control law $u=\K(x_t,\dot{x}_t,\ddot{x}_t)$ that depends on $\ddot{x}_t$. 
Then the closed-loop dynamics would become an AFDE.

\newpage
%%%%%%%%%%%%%%%%%%%%%%%%%%%%%%%%%%%%%%%%%%%%%%%%%%%%%%%%%%%%%%%%%%%%%%%%%%%%%%%%%%%%%%%%%%%%
\section{Case-study: regulated delayed predator-prey model}
\label{sec:predatorprey}
%%%%%%%%%%%%%%%%%%%%%%%%%%%%%%%%%%%%%%%%%%%%%%%%%%%%%%%%%%%%%%%%%%%%%%%%%%%%%%%%%%%%%%%%%%%%

\begin{center}
\begin{table}[!b]
\centering
\begin{tabular*}{300pt}{@{\extracolsep\fill}lll@{\extracolsep\fill}}%
\toprule
Description & Parameter & Value \\
\midrule
prey growth rate & $r$ &  1\\
prey self-regulation rate & $a$ & 1  \\
predation rate of the prey &  $p$ &  4  \\
conversion rate of prey into predator  & $b$ &  1.2  \\
predator intraspecific competition & $m$ &  0.1 \\
predator mortality rate &  $d$ & 1 \\
predator maturation time &  $\tau$ & 5 \\
\midrule
prey lower limit &  $ x_{1, {\rm min}}$ &   0.05 \\
prey upper limit &  $ x_{1, {\rm max}}$ &   0.6 \\
\bottomrule
\end{tabular*}
\caption{Parameters of the predator-prey model.}\label{tab:predator_prey_param}%
\end{table}
\end{center}

Finally, we investigate a predator-prey problem which is subject to time delay \cite{macdonald1978,kuang1993delay,stepan1986great}. This application is becoming more and more important considering the fragile ecosystems as a consequence of climate change.
The evolution of predator and prey populations is described in an ecosystem, and human intervention, regarded as control input, is used to control these numbers \cite{lenhart2007optimal}.
We use the proposed safety-critical control framework to regulate the numbers of predators or preys and maintain these numbers within safe bounds.
Naturally, such ecosystems contain significant time delays because each predator and prey has a finite maturation period in its life before they start interacting with each other.
While the underlying delayed dynamics have been analysed extensively during the last few decades by many researchers \cite{ruan2009nonlinear}, to the best of our knowledge, safety-critical control has not yet been applied to address this problem due to the lack of theoretical background that endows time delay systems with provable guarantees of safety.
Now we use this problem to demonstrate that our proposed framework is able to provide the desired safety guarantees for systems with state delays.

We denote the population of the preys by $x_1$ and the population of the predators by $x_2$, and we model their dynamic interaction by the following nondimensionalized system:
\begin{equation} \label{eq:predator_prey}
\begin{split}
\underbrace{
\begin{bmatrix}
\dot{x}_1(t) \\
\dot{x}_2(t)
\end{bmatrix}
}_{
\dot{x}(t)}
=
\underbrace{
\begin{bmatrix}
r x_1(t)-a x_1^2(t) -px_1(t)x_2(t)\\
b p x_1(t-\tau)x_2(t-\tau)-d x_2(t) - m x_2^2(t)
\end{bmatrix}}_{f(x(t),x(t-\tau))}
+
\underbrace{
\begin{bmatrix}
0 \\
1
\end{bmatrix}
}_{g(x(t))}u(t),
\end{split}
\end{equation}
where the time delay $\tau$ indicates the maturation time of the predators, $r$ describes the growth rate of the prey in the absence of predators, $a$ denotes the self-regulation constant of the prey, $p$ describes the predation rate of the prey by predators, $b$ indicates the rate of conversion of consumed prey to predator, $d$ is the specific mortality of predator in the absence of prey, $m$ describes the intraspecific competition among predators \cite{gourley2004stage},
and $u$ quantifies the effect of human intervention affecting the number of predators.

Without control input (${u(t) \equiv 0}$), system \eqref{eq:predator_prey} has four equilibria:
${(x_1, x_2) = (0, 0)}$,
${(x_1, x_2) = (0, -d/m)}$,
${(x_1, x_2) = (r/a, 0)}$ and
${(x_1, x_2) = (mr+pd, bpr-ad)/(am+bp^2)}$,
from which the last one is relevant. It corresponds to a constant population in the ecosystem, which is independent of the time delay. However, its stability depends on the delay, and there exists a critical time delay above which the equilibrium becomes unstable by a supercritical Hopf bifurcation that induces stable periodic solutions \cite{stepan1986great}. In this case, from the biological point of view, the populations of predators and preys are oscillating.

%%%%%%%%%%%%%%%%%%%%%%%%%%%%%%%%%%%%%%%%
\begin{figure}[h!]
\centerline{\includegraphics[width=0.899\columnwidth]{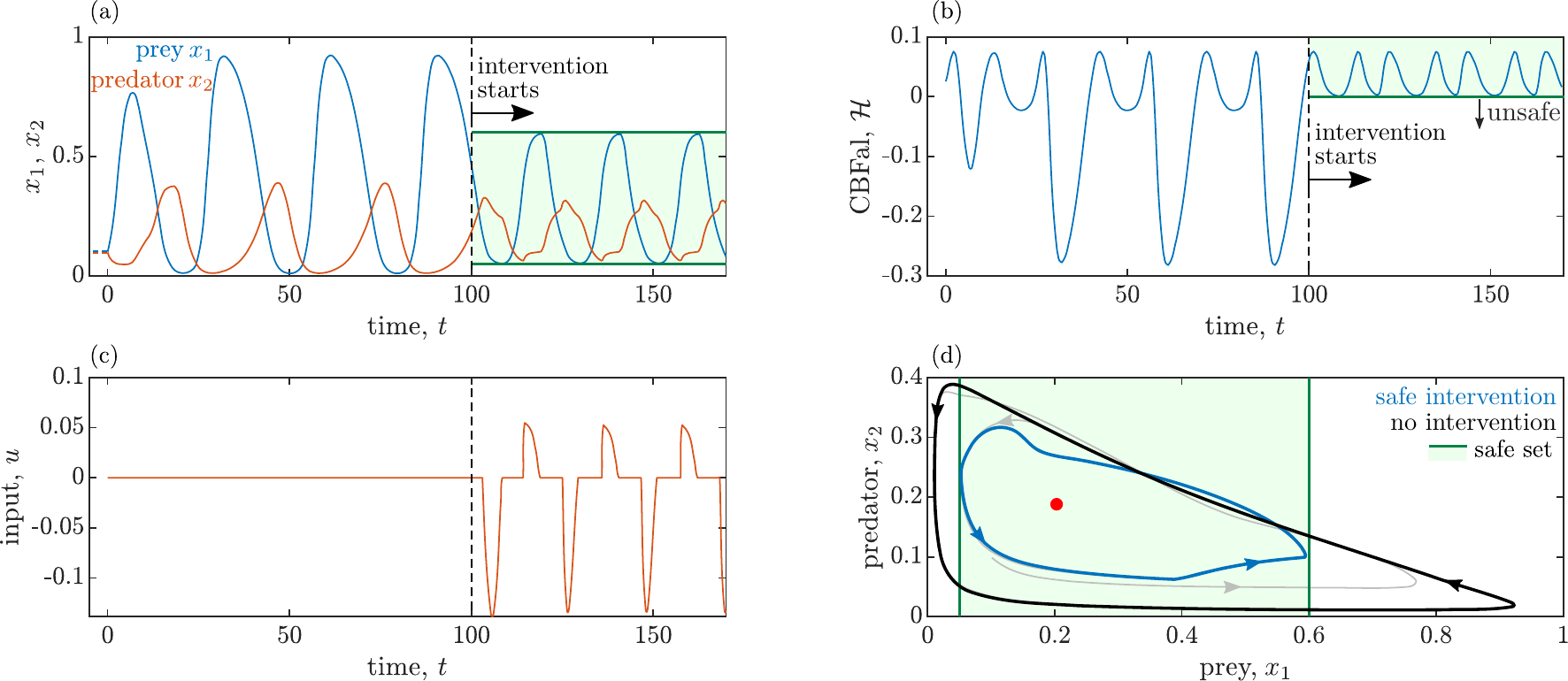}}
\caption{Safety-critical control of the delayed predator-prey model \eqref{eq:predator_prey} to keep the number of preys within limits.}
\label{fig:predator_prey}
\end{figure}
%%%%%%%%%%%%%%%%%%%%%%%%%%%%%%%%%%%%%%%%

In the literature there exist several methods to regulate the predator-prey system such as the addition of food, pesticide or insecticide  \cite{san1974optimal,lenhart2007optimal}.
In this example, we directly control the number of predators by capturing or releasing them, i.e., the control input $u$ enters the dynamics of the predator population $x_2$.

We seek to synthesize a controller that keeps the ecosystem safe.
Specifically, our aim is to regulate the number $x_1$ of preys and keep their population within prescribed bounds.
On one hand, we specify a lower limit $x_{1, {\rm min}}$ to avoid getting close to the danger of prey extinction.
On the other hand, we include an upper limit $x_{1, {\rm max}}$, to prevent the number of preys from increasing too much, which could result in an unsustainable ecosystem with over-consumption of available food resources.

Thus, we use the candidate control barrier functional:
\begin{equation}
\begin{split}
\mathcal{H}(x_t)&=-\big(x_1(t) - x_{1, {\rm min}}\big)\big(x_1(t) - x_{1, {\rm max}}\big),\\
\end{split}
\end{equation}
and we seek to maintain $\mathcal{H}(x_t)\geq 0$ for all time.
The time derivative of $\mathcal{H}$ along \eqref{eq:predator_prey} reads:
\begin{equation}
\dot{\mathcal{H}}(x_t) = 2 \left( \bar{x}_1 - x_1(t) \right) \left( r x_1(t)-a x_1^2(t) -px_1(t)x_2(t) \right)
\end{equation}
with ${\bar{x}_1 = (x_{1, {\rm min}} + x_{1, {\rm max}}) / 2}$.
Since $\dot{\mathcal{H}}$ only depends on $x_t$ (i.e., the corresponding $\mathcal{L}_\F \mathcal{H}$ does not contain terms of $\dot{x}_t$ and $\mathcal{L}_\G \mathcal{H}(x_t) = 0$), we construct the extended CBFal \eqref{eq:ext_safety_func} with $\alpha(r) = \gamma r$, $\gamma > 0$, that becomes:
\begin{equation}
\mathcal{H}_{\rm e}(x_t) = 2 \left( \bar{x}_1 - x_1(t) \right) \left( r x_1(t)-a x_1^2(t) -px_1(t)x_2(t) \right) - \gamma (x_1(t) - x_{1, {\rm min}})(x_1(t) - x_{1, {\rm max}}).
\end{equation}
Its derivative,
\begin{multline}
\dot{\mathcal{H}}_{\rm e}(x_t,\dot{x}_t,u)
= - 2 \left( r x_1(t)-a x_1^2(t) -px_1(t)x_2(t) \right)^2
+ 2 \left( \bar{x}_1 - x_1(t) \right) \left( r - 2 a x_1(t) -p x_2(t) \right) \left( r x_1(t)-a x_1^2(t) -px_1(t)x_2(t) \right) \\
- 2 p \left( \bar{x}_1 - x_1(t) \right) x_1(t) \left( b p x_1(t-\tau)x_2(t-\tau)-d x_2(t) - m x_2^2(t) + u \right)
+ 2 \gamma \left( \bar{x}_1 - x_1(t) \right) \left( r x_1(t)-a x_1^2(t) -px_1(t)x_2(t) \right),
\end{multline}
depends on the input $u$.
Thereby safety can be enforced by synthesizing controllers according to Theorem~\ref{thm:safety_delay_extended}.
We implement the QP-based controller~\eqref{eq:min-norm_extended}, with desired controller ${\K_{\rm des}(x_t,\dot{x}_t)=0}$ and linear class-$\Keinf$ function
${\alpha_{\rm e}(r) = \gamma_{\rm e} r}$, $\gamma_{\rm e}>0$.

Simulation results are illustrated in Fig.~\ref{fig:predator_prey} for the model parameters are listed in Table~\ref{tab:predator_prey_param}, ${\gamma=1}$, ${\gamma_{\rm e}=1}$ and initial conditions ${x_t(\vartheta)=\begin{bmatrix}0.1 & 0.1\end{bmatrix}^\top}$, ${\vartheta \in [-\tau,0]}$.
First, we demonstrate the natural evolution of system \eqref{eq:predator_prey} without active intervention (${u(t) \equiv 0}$) during the time interval ${t \in [0, 100]}$.
Then, the safety-critical controller is turned on after ${t = 100}$.

Without safety-critical control, the trajectory of the system converges to a stable periodic orbit with large fluctuation in the populations (see black orbit in panel (d)).
Notice that the prey population goes outside the range ${[x_{1, {\rm min}},x_{1, {\rm max}}]}$, and the system periodically leaves the safe set as indicated by the negative values of the functional $\mathcal{H}$ during ${t \in [0, 100]}$.

With safety-critical control, the level of intervention is quantified in panel (c).
To keep the number of preys within the prescribed limits, predators are captured (${u(t)<0}$) when the population of preys is small and predators are added (${u(t)>0}$) when the prey population is too large.
The controller intervenes minimally in the system, only when the prey population is too close to the safe limits, and otherwise lets the predator and prey populations evolve naturally (i.e., ${u(t)=0}$ for some duration of time).
The effect of this intervention on safety is clearly seen: it maintains nonnegative values for the functional $\mathcal{H}$, keeps the number of preys within safe bounds, and forces the system to evolve within the safe set (green shaded areas in panels (a) and (d)).
In fact, the solution converges to a stable limit cycle (blue orbit in panel (d)) located inside the safe set (and grazing the boundaries of the safe set periodically).
Compared to the uncontrolled scenario, the proposed controller modifies the size of the stable limit cycle around the unstable equilibrium without changing these stability properties.
Remarkably, such desired behavior is generated by a systematic controller synthesis procedure based on control barrier functional theory.

%%%%%%%%%%%%%%%%%%%%%%%%%%%%%%%%%%%%%%%%%%%%%%%%%%%%%%%%%%%%%%%%%%%%%%%%%%%%%%%%%%%%%%%%%%%%
\section{Conclusion}\label{sec:conclusion}
%%%%%%%%%%%%%%%%%%%%%%%%%%%%%%%%%%%%%%%%%%%%%%%%%%%%%%%%%%%%%%%%%%%%%%%%%%%%%%%%%%%%%%%%%%%%
In this work, we have discussed the safety of time delay systems that include state delays.
First, we have proposed a method to formally certify the safety of autonomous delayed dynamical systems by means of {\em safety functionals} defined over the infinite dimensional state space.
We have broken down how the required time derivative of the potentially complicated nonlinear safety functional can be expressed.
% as a linear functional according to the Riesz representation theorem.
Second, we have extended this theory to control systems with time delay in order synthesize safety-critical controllers that are endowed with rigorous safety guarantees.
The essence of the proposed method is the introduction of {\em control barrier functionals}
% with a large set of available control inputs.
that guide the selection of safe control inputs for time delay systems in a similar fashion to how control barrier functions yield safety for delay-free systems.
% It creates safety-critical controller synthesis to systematically find pointwise optimal intervention.
We have provided formal safety guarantees with their proofs.
Third, we have also incorporated control barrier functionals into optimization problems to find pointwise optimal safe controllers, and we have demonstrated the applicability of the proposed method on illustrative examples.

The proposed theoretical framework open ways
% to efficient control optimization strategies crossing the border
towards provably safe and reliable operation of time delay systems.
While the paper presents some demonstrative examples, we believe that this approach could be used in a much wider range of application fields.
Our future goals are to implement this method in engineering applications and to perform experiments to validate the theoretical results.

\section*{Acknowledgments}
The research reported in this paper has been supported by the National Research, Development and Innovation Fund (TKP2020 NC, Grant No.\ BME-NCS and TKP2021, Project no.\ BME-NVA-02) under the auspices of the Ministry for Innovation and Technology of Hungary, 
by the National Science Foundation (CPS Award \#1932091), by Aerovironment and by Dow (\#227027AT).

\subsection*{Author contributions}

A. K. Kiss developed the theory and conducted the numerical simulations.
T. G. Molnar supported the theory and helped with the writing.
A. D. Ames and G. Orosz supervised the project, including the development of the theory, results and writing.

\subsection*{Conflict of interest}

The authors declare no potential conflict of interests.

\subsection*{Data availability statement}
Data sharing not applicable to this article as no datasets were generated or analyzed during the current study.

\bibliography{main_JNRC}

\appendix

%%%%%%%%%%%%%%%%%%%%%%%%%%%%%%%%%%%%%%%%%%%%%%%%%%%%%%%%%%%%%%%%%%%%%%%%%%%%%%%%%%%%%%%%%%%%
\section{Derivatives of Functionals}\label{sec:appdx_Gateaux}
%%%%%%%%%%%%%%%%%%%%%%%%%%%%%%%%%%%%%%%%%%%%%%%%%%%%%%%%%%%%%%%%%%%%%%%%%%%%%%%%%%%%%%%%%%%%

This appendix shows the necessary definitions and lemma for the proof of Theorem~\ref{thm:H_func_deriv}.

Recall that for the delay-free system in Section~\ref{sec:delay_free}, the time derivative of function $h$ along the system can be considered as the directional derivative of $h$ along $\dot{x}$. This is generalized to functionals in that the time derivative of functional $\mathcal{H}$ is the directional derivative of $\mathcal{H}$ along $\dot{x}_t$, which is given by the so-called G\^{a}teaux derivative formulated as follows.
% Since the G\^{a}teaux derivative is the functional analysis generalization of the directional derivative, $\dot{\mathcal{H}}(x_t)$ becomes the G\^{a}teaux derivative of the functional $\mathcal{H} \colon \B \to \R$ in the direction of $\dot{x}_t \in \B$ evaluated at $x_t \in \B$, defined by
\begin{definition}[\textbf{G\^{a}teaux Derivative}]
% \textit{
Functional $\mathcal{H} \colon \B \to \R$ is G\^{a}teaux differentiable at $x_t \in \B$ if there exists a functional $D_{\rm G}\mathcal{H}(x_t) \colon \Q \to \R$ such that $\forall \phi \in \Q$:
\begin{equation}\label{eq:gateaux_definition}
D_{\rm G}\mathcal{H}(x_t)(\phi)=\lim_{\Delta t \to 0} \frac{\mathcal{H}(x_t + \Delta t \phi)-\mathcal{H}(x_t)}{\Delta t} , \quad \Delta t \in \R,
\end{equation}
where $D_{\rm G}\mathcal{H}(x_t)$ is called the \textbf{G\^{a}teaux derivative} of $\mathcal{H}$ at $x_t$, that is evaluated along $\phi$.
% }
\end{definition}
Similarly, the generalization of the gradient $\nabla h$ is the so-called Fr\'{e}chet derivative of $\mathcal{H}$, defined by the implicit form below.
\begin{definition}[\textbf{Fr\'{e}chet derivative}]
% \textit{
Functional $\mathcal{H} \colon \B \to \R$ is Fr\'{e}chet differentiable at $x_t \in \B$ if there exists a bounded linear functional $D_{\rm F}\mathcal{H}(x_t) \colon \Q \to \R$ such that:
\begin{equation}\label{eq:frechet_definition}
    \lim_{\|\phi\| \to 0} \frac{|\mathcal{H}(x_t+\phi) - \mathcal{H}(x_t) - D_{\rm F}\mathcal{H}(x_t)\phi |}{\|\phi\|} = 0, \quad \phi \in \Q,
\end{equation}
where $D_{\rm F}\mathcal{H}(x_t)$ is called the \textbf{Fr\'{e}chet derivative} of $\mathcal{H}$ at $x_t$.
% }
\end{definition}

% frechet , vesszük egy gradiens hát és még rá is szorzunk, amin hattatjuk.

If $\mathcal{H}$ is Fr\'{e}chet differentiable, it implies that it is also G\^{a}teaux differentiable and directional derivatives exist in all directions.
The following lemma formally establishes this connection between the G\^{a}teaux and Fr\'{e}chet derivatives to prove Theorem~\ref{thm:H_func_deriv}.
\begin{lemma}\cite{andrews2010ricci}\label{lem:gateaux_is_frechet}
\textit{
If $\mathcal{H} \colon \B \to \R$ is Fr\'{e}chet differentiable at $x_t \in \B$, then it is also G\^{a}teaux differentiable at $x_t$, and the G\^{a}teaux derivative $D_{\rm G}\mathcal{H}(x_t) \colon \Q \to \R$ is given by a bounded linear functional that is the  Fr\'{e}chet derivative $D_{\rm F}\mathcal{H}(x_t) \colon \Q \to \R$:
\begin{equation}\label{eq:gateaux_is_frechet}
D_{\rm G}\mathcal{H}(x_t)(\dot{x}_t)=D_{\rm F}\mathcal{H}(x_t)\dot{x}_t, \quad \forall \dot{x}_t \in \Q.
\end{equation}
}
\end{lemma}

The proof can be found in \cite{andrews2010ricci}, Proposition A.3.
Consequently, if $\mathcal{H}$ is continuously Fr\'{e}chet differentiable, its Fr\'{e}chet and G\^{a}teaux derivatives are continuous linear functionals.
While nonlinear functionals have no general form, continuous linear functionals can be represented in the form of Stieltjes integrals, as provided by the Riesz representation theorem.
% Note that the converse is not necessarily true, that is $$ if and in general is not necessarily a linear operator.

\begin{theorem}[\textbf{Riesz Representation Theorem \cite{riesz1955functional,hale1977theory,diekmann2012delay,breda2014stability,Stepan1989}}]\label{thm:riesz}
\textit{
For every  continuous linear functional $\mathcal{L} \colon \B \to \R$ there exists a unique function ${\eta \colon [-\tau,0] \to \R^{1 \times n}}$ that is of bounded variation such that $\forall \phi \in \B$:
\begin{equation}
\mathcal{L}(\phi)
% =\langle  \phi,  \eta \rangle ,
% \end{equation}
% where $\langle  \,.\,,\,.\, \rangle $ denotes the inner product
% \begin{equation}\label{eq:inner_product}
% \langle  \phi,  \eta \rangle
= \int_{-\tau}^{0} {\rm d}_\vartheta \eta(\vartheta) \phi(\vartheta),
\end{equation}
where the integral is a Stieltjes type and ${\rm d}_\vartheta \eta(\vartheta)$ is interpreted as a measure corresponding to the function $\eta$.
}
\end{theorem}

The proof can be found in \cite{riesz1955functional}.
Accordingly, the Fr\'{e}chet derivative of a continuously Fr\'{e}chet differentiable functional ${\mathcal{H} \colon \B \to \R}$ is a continuous linear functional, that can be represented as follows when evaluated along $\dot{x}_t$:
\begin{equation}\label{eq:Riesz_rep}
% \dot{\mathcal{H}}(x_t,\dot{x}_t)
D_{\rm F}\mathcal{H}(x_t)\dot{x}_t = \int_{-\tau}^{0} {\rm d}_\vartheta \eta(x_t,\vartheta) \dot{x}_t(\vartheta).
% = \int_{-\tau}^{0}  w(x_t,\vartheta) \dot{x}_t(\vartheta) {\rm d}\vartheta,
\end{equation}
Theorem~\ref{thm:H_func_deriv} is therefore a direct consequence of Lemma~\ref{lem:gateaux_is_frechet} and Theorem~\ref{thm:riesz}.
\section{Example functional with double integral}\label{sec:appdx_double_int}
%%%%%%%%%%%%%%%%%%%%%%%%%%%%%%%%%%%%%%%%%%%%%%%%%%%%%%%%%%%%%%%%%%%%%%%%%%%%%%%%%%%%%%%%%%%%

Let us consider the system \eqref{eq:Aut_FDE} with the functional $\mathcal{H}$ defined as:
\begin{equation}\label{eq:functional}
\begin{split}
 \mathcal{H}(x_t) = h\bigg(  x_t(0), \int_{-\tau}^0 \rho(\vartheta)\, \kappa \big( x_t(\vartheta) \big)  {\rm d} \vartheta,
\int_{-\tau}^0 \int_{-\tau}^0 \omega( \vartheta,\chi)\, \Big( \mu \big( x_t(\vartheta) \big) \circ \nu \big( x_t(\chi) \big) \Big) {\rm d}\vartheta {\rm d}\chi  \bigg),
\end{split}
\end{equation}
cf.~\eqref{eq:safe_funnal_gen_exa}, where ${h \colon \R^n \times \R^n \times \R^n \to \R}$,  ${\kappa,\mu,\nu \colon \R^n \to \R^n}$ are nonlinear functions, $\circ$ refers to element-wise multiplication, while the density functions ${\rho \colon [-\tau,0] \to \R^{n\times n}}$ and ${\omega \colon [-\tau,0] \times [-\tau,0] \to \R^{n\times n}}$ are assumed to be continuously differentiable and these functions allow us take into account the states between time moments $t-\tau$ and $t$. 

One may take the time derivative of $\mathcal{H}$ in \eqref{eq:functional}, substitute \eqref{eq:state_derivative}, and use
\begin{align}
\begin{split}
\frac{\partial}{\partial t} \kappa \big( x_t(\vartheta) \big) & = \frac{\partial}{\partial \vartheta} \kappa \big( x_t(\vartheta) \big), \\
\frac{\partial}{\partial t} \mu \big( x_t(\vartheta) \big) & = \frac{\partial}{\partial \vartheta} \mu \big( x_t(\vartheta) \big), \\
\frac{\partial}{\partial t} \nu \big( x_t(\chi) \big) & = \frac{\partial}{\partial \chi} \nu \big( x_t(\chi) \big).
\end{split}
\end{align}
By executing partial integration, these steps yield:
\begin{equation}\label{eq:functionalder}
\begin{split}
\dot{\mathcal{H}}(x_t,\dot{x}_t) &= \nabla_1 h(\dots)  \F(x_t)
\\
&+\nabla_2 h(\dots)  \bigg( \rho( 0) \kappa \big( x(t) \big)  
- \rho( -\tau) \kappa \big( x(t-\tau) \big) 
% \\ & \qquad \qquad
- \int_{-\tau}^0
% \nabla \rho(\vartheta)\, 
\rho'(\vartheta)\,
\kappa \big( x(t+\vartheta) \big)  {\rm d} \vartheta \bigg)
\\
&+\nabla_3 h(\dots)  \bigg( \int_{-\tau}^0 \omega( 0,\chi)\, \Big( \mu \big( x(t) \big) \circ \nu \big(x(t+\chi) \big) \Big) 
% \\ &\qquad \qquad 
- \omega( -\tau,\chi)\, \Big( \mu \big( x(t-\tau) \big) \circ  \nu \big(x(t+\chi) \big) \Big) {\rm d}\chi
\\
&\qquad \quad 
+ \int_{-\tau}^0 \omega( \vartheta,0)\, \Big( \mu \big( x(t+\vartheta) \big) \circ \nu \big(x(t) \big) \Big) 
% \\ &\qquad \qquad
- \omega( \vartheta,-\tau)\, \Big( \mu \big( x(t+\vartheta) \big) \circ \nu \big(x(t-\tau) \big) \Big) {\rm d}\vartheta
\\
&\qquad \quad -\int_{-\tau}^0 \int_{-\tau}^0 
\Big(
% \nabla_\vartheta 
\frac{\partial}{\partial \vartheta}
\omega( \vartheta,\chi)  +
% \nabla_\chi
\frac{\partial}{\partial \vartheta}
\omega( \vartheta,\chi) \Big)
% \\ &\qquad \qquad 
% \times
\Big( \mu \big( x(t+\vartheta) \big) \circ \nu \big(x(t+\chi) \big) \Big) {\rm d}\vartheta {\rm d}\chi  \bigg),
\end{split}
\end{equation}
cf.~\eqref{eq:safe_funnal_der_exa}, where $\nabla_j$ represents the gradient with respect to the $j$-th vector-valued variable, and $(\dots)$ is a shorthand notation for evaluation at the argument of $h$ as in \eqref{eq:functional}.
Note that terms of $\dot{x}_t$ drop from $\dot{\mathcal{H}}$ similar to the case in Remark~\ref{rem:simplification}.

%%%%%%%%%%%%%%%%%%%%%%%%%%%%%%%%%%%%%%%%%%%%%%%%%%%%%%%%%%%%%%%%%%%%%%%%%%%%%%%%%%%%%%%%%%%%
\section{KKT Conditions}\label{sec:appdx_KKT}
%%%%%%%%%%%%%%%%%%%%%%%%%%%%%%%%%%%%%%%%%%%%%%%%%%%%%%%%%%%%%%%%%%%%%%%%%%%%%%%%%%%%%%%%%%%%

Here, we briefly discuss the derivation steps for determining the solution \eqref{eq:min-norm} to the quadratic program~\eqref{eq:QP_delay} in Corollary~\ref{cor:min_norm}.
% \begin{proof}[Proof of Theorem]

Let us define $\Delta \K \colon \B \times \Q \to \R^m$,
${\Delta \K(x_t,\dot{x}_t) = \K(x_t,\dot{x}_t) - \K_{\rm des}(x_t,\dot{x}_t)}$
% ${\Delta u = u - u_{\rm des}}$
and consider the expressions of $\dot{\mathcal{H}}$ in~\eqref{eq:H_Lie_deriv} and $\phi$, $\phi_{0}$ in Corollary~\ref{cor:min_norm}.
% in~\eqref{eq:PHI_PHI0} with ${\phi_{0} \neq 0}$.
Then, we can restate~\eqref{eq:QP_delay} as
\begin{align}\label{eq:QP_proof}
\begin{split}
\K(x_t,\dot{x}_t) & = \K_{\rm des}(x_t,\dot{x}_t) + \Delta \K(x_t,\dot{x}_t), \\
\Delta \K(x_t,\dot{x}_t) & =
\underset{\Delta u \in \R^m}{\operatorname{argmin}} \quad \frac{1}{2}\| \Delta u \|_2^2  \\
& \qquad\quad \mathrm{s.t.} \quad \phi(x_t, \dot{x}_t) + \phi_{0}(x_t,\dot{x}_t) \Delta u \geq 0.
\end{split}
\end{align}
% where
% \begin{equation}\label{}
% \begin{split}
%  \phi(x_t, \dot{x}_t) & = \mathcal{L}_\F  \mathcal{H}  + \mathcal{L}_\G  \mathcal{H} \, u_{\rm des} + \alpha( \mathcal{H} )\\
%  \phi_{0}(x_t) & = \mathcal{L}_\G  \mathcal{H}
% \end{split}
% \end{equation}

% $\phi = L _\F  \mathcal{H}  + \mathcal{L}_\G  \mathcal{H} \, u_{\rm des} + \alpha( \mathcal{H} )$ and $\phi_0 = L _\G  \mathcal{H}$.

In order to solve \eqref{eq:QP_proof}, let us define the \textit{Lagrangian} $L \colon \R^m \times \B \times \Q \to \R$ associated with the optimization problem \eqref{eq:QP_proof} as:
\begin{equation}
L(\Delta u, x_t, \dot{x}_t) = %\frac{1}{2}\Delta u^\top \Delta u 
\|\Delta u(x_t,\dot{x}_t) \|_2^2- \mu(x_t,\dot{x}_t) \big(\phi(x_t,\dot{x}_t) + \phi_{0}(x_t,\dot{x}_t) \Delta u \big),
\end{equation}
where $\mu \colon \B \times \Q \to \R$ is the Lagrange multiplier associated with the inequality constraint.
This optimization problem has convex objective and affine constraint, hence the {\em Karush-Kuhn-Tucker (KKT) conditions}~\cite{Boyd2004} provide the necessary and sufficient conditions for optimality, listed as
% The KKT optimality conditions imply that there exists a Lagrange multiplier ${\mu : X \to \R}$ such that $\mu(x)$ and $\Delta k(x)$ satisfy
\begin{align}
% \begin{split}
& \mu(x_t,\dot{x}_t) \geq 0, & \text{Dual Feasibility}
\label{eq:kkt_dual_feasibility} \\
& \Delta \K(x_t,\dot{x}_t) = \mu(x_t,\dot{x}_t) \phi_{0}^\top(x_t,\dot{x}_t), & \text{Stationary}
\label{eq:kkt_stationary} \\
& \phi(x_t,\dot{x}_t) + \phi_{0}(x_t,\dot{x}_t) \Delta \K(x_t,\dot{x}_t) \geq 0, & \text{Primal Feasibility}
\label{eq:kkt_primal_feasibility} \\
& \mu(x_t,\dot{x}_t) \big(\phi(x_t,\dot{x}_t) + \phi_{0}(x_t,\dot{x}_t) \Delta \K(x_t,\dot{x}_t) \big) = 0. & \text{Complementary Slackness}
\label{eq:kkt_complementary_slackness}
% \end{split}
\end{align}

We decompose the dual feasibility condition~\eqref{eq:kkt_dual_feasibility} into two cases: ${\mu(x_t,\dot{x}_t) = 0}$ and ${\mu(x_t,\dot{x}_t) > 0}$.
If $\mu(x_t,\dot{x}_t) = 0$, then the stationary condition~\eqref{eq:kkt_stationary} leads to:
\begin{equation}
\Delta \K(x_t,\dot{x}_t) = 0,
\end{equation}
while the primal feasibility condition~\eqref{eq:kkt_primal_feasibility} implies $\phi(x_t,\dot{x}_t) \geq 0$.
% If $ \phi(x_t,\dot{x}_t) + \phi_{0}(x_t,\dot{x}_t) \Delta \K(x_t,\dot{x}_t) > 0$, then the complementary slackness condition \eqref{eq:kkt_complementary_slackness} gives
% \begin{equation}
% \mu(x_t,\dot{x}_t) = 0.
% \end{equation}
% With the stationary condition~\eqref{eq:kkt_stationary} this leads to
% \begin{equation}
% \Delta \K(x_t,\dot{x}_t) = 0,
% \end{equation}
% which implies $\phi(x_t,\dot{x}_t) > 0$ and that there is no change in the desired control input, $\K(x_t,\dot{x}_t)=\K_{\rm des}(x_t,\dot{x}_t)$.
% Substituting back to the primal feasibility condition \eqref{eq:kkt_primal_feasibility} and 
If  ${\mu(x_t,\dot{x}_t) > 0}$, then the complementary slackness condition \eqref{eq:kkt_complementary_slackness} yields:
\begin{equation}
\phi(x_t,\dot{x}_t) + \phi_{0}(x_t,\dot{x}_t) \Delta \K(x_t,\dot{x}_t) = 0,
\end{equation}
and we can express $\Delta \K$ as:
\begin{equation}
\Delta \K(x_t,\dot{x}_t) = -\frac{\phi(x_t,\dot{x}_t) \phi_0^\top(x_t,\dot{x}_t)}{\phi_0(x_t,\dot{x}_t) \phi_0^\top(x_t,\dot{x}_t)}.
\end{equation}
With the stationary condition \eqref{eq:kkt_stationary} this also means $\mu(x_t,\dot{x}_t) = -\phi(x_t,\dot{x}_t)/(\phi_0(x_t,\dot{x}_t) \phi_0^\top(x_t,\dot{x}_t))$ that implies $\phi(x_t,\dot{x}_t) < 0$ since ${\mu(x_t,\dot{x}_t) > 0}$.

Then the closed-form solution of the quadratic program~\eqref{eq:QP_delay} can be written as:
\begin{equation}\label{eq:QP_solu}
\K(x_t,\dot{x}_t)= 
\begin{cases}
\K_{\rm des}(x_t,\dot{x}_t) & {\rm if} \; \phi(x_t,\dot{x}_t) \geq 0,\\
\K_{\rm des}(x_t,\dot{x}_t) -\frac{\phi(x_t,\dot{x}_t) \phi_0^\top(x_t,\dot{x}_t)}{\phi_0(x_t,\dot{x}_t) \phi_0^\top(x_t,\dot{x}_t)}  & {\rm otherwise},
\end{cases}
\end{equation}
cf.~\eqref{eq:min-norm}.
% This completes the proof of the theorem. 
% \end{proof}
Note that the CBFal condition \eqref{eq:CBFal_condition} is equivalent to $\mathcal{L}_\G \mathcal{H}(x_t,\dot{x}_t) = 0 \implies \mathcal{L}_\F \mathcal{H}(x_t,\dot{x}_t) + \alpha \big( \mathcal{H}(x_t) \big) > 0$, and we have $\phi_0(x_t,\dot{x}_t) = 0 \implies \phi(x_t,\dot{x}_t) > 0$.
Hence division by zero cannot occur in \eqref{eq:QP_solu} if $\mathcal{H}$ is a CBFal.

% \nocite{*}% Show all bib entries - both cited and uncited; comment this line to view only cited bib entries;
% \bibliography{wileyNJD-AMA}%

% \clearpage

% \section*{Author Biography}

% \begin{biography}{\includegraphics[width=66pt,height=86pt,draft]{empty}}{\textbf{Author Name.} This is sample author biography text this is sample author biography text this is sample author biography text this is sample author biography text this is sample author biography text this is sample author biography text this is sample author biography text this is sample author biography text this is sample author biography text this is sample author biography text this is sample author biography text this is sample author biography text this is sample author biography text this is sample author biography text this is sample author biography text this is sample author biography text this is sample author biography text this is sample author biography text this is sample author biography text this is sample author biography text this is sample author biography text.}
% \end{biography}

\end{document}